\documentclass[]{interact}

\usepackage{epstopdf}
\usepackage[caption=false]{subfig}

\usepackage[numbers,sort&compress]{natbib}
\bibpunct[, ]{[}{]}{,}{n}{,}{,}
\renewcommand\bibfont{\fontsize{10}{12}\selectfont}
\usepackage{xcolor}
\theoremstyle{plain}
\newtheorem{theorem}{Theorem}[section]
\newtheorem{lemma}[theorem]{Lemma}

\theoremstyle{definition}
\newtheorem{definition}[theorem]{Definition}

\theoremstyle{remark}
\newtheorem{remark}{Remark}

\begin{document}

\articletype{}

\title{Simultaneous Estimation of Piecewise Constant Coefficients in Elliptic PDEs via Bayesian Level-Set Methods }

\author{
\name{Anuj Abhishek\textsuperscript{a}\thanks{CONTACT Taufiquar Khan. Email: taufiquar.khan@charlotte.edu} and Thilo Strauss\textsuperscript{b} and Taufiquar Khan\textsuperscript{c}}
\affil{\textsuperscript{a}Department of Mathematics, Applied Mathematics and Statistics, Case Western Reserve University, OH, U.S.A.; \textsuperscript{b}Xi'an Jiaotong-Liverpool University, School of AI and advanced computing,  Jiangshu Province, China;\textsuperscript{c}Department of Mathematics and Statistics and Center for TAIMing AI, University of North Carolina at Charlotte, NC, U.S.A.}
}

\maketitle

\begin{abstract}
In this article, we propose a non-parametric Bayesian level-set method for simultaneous reconstruction of two different piecewise constant coefficients in an elliptic partial differential equation. We show that the Bayesian formulation of the corresponding inverse problem is well-posed and that the posterior measure as a solution to the inverse problem satisfies a Lipschitz estimate with respect to the measured data in terms of Hellinger distance. We reduce the problem to a shape-reconstruction problem and use level-set priors for the parameters of interest. We demonstrate the efficacy of the proposed method using numerical simulations by performing reconstructions of the original phantom using two reconstruction methods. Posing the inverse problem in a Bayesian paradigm allows us to do statistical inference for the parameters of interest, whereby we are able to quantify the uncertainty in the reconstructions for both methods. This illustrates a key advantage of Bayesian methods over traditional algorithms.
\end{abstract}

\begin{amscode}
    62G05,  35J15, 62F15
\end{amscode}

\section{Introduction}
{In this article, we formulate and study a two-parameter coefficient identification problem for elliptic PDEs in a Bayesian framework. The motivation for studying this problem comes from an inverse problem for Diffuse Optical Tomography (DOT).} In the inverse problem for DOT, one seeks to recover the spatially varying diffusion and absorption parameters of the medium from measured photon density on the boundary. 
For details on the use of DOT in the field of medical imaging, we refer the reader to \cite{arr_99,arr_scho,boas_01} and the references therein. The problem we describe in the paragraphs below is related to the inverse problem for DOT.
\par For the benefit of the readers, we first begin with a detailed self-contained description of the inverse problem at hand. We will also review the main ideas from \cite{stuart_10} which will illustrate the non-parametric Bayesian approach to inverse problems as used in this article.
\subsection{A joint coefficient-identification problem}

\par We consider the following joint coefficient-identification problem. We want to determine simultaneously the unknown coefficients $a(x)$ and $b(x)$ in the following elliptic PDE:
\begin{align}\label{eq:1.1}
  -\nabla \cdot \left( a(x) \nabla u (x)\right) + b(x) u(x) = 0 \quad \mbox{ in } \Omega . 
\end{align}
 from prescribed Neumann data i.e., $-a\frac{\partial u}{\partial \nu}|_{\partial \Omega}$, and Robin data i.e., $(u(x)+2 a\frac{\partial u}{\partial \nu})|_{\partial \Omega}$, pairs on the boundary $\partial \Omega$. In a seminal paper \cite{Harr09}, Harrach showed that if $a(x)$ is piecewise constant and $b(x)$ is piecewise analytic, then unique recovery of both the coefficients is possible from Dirichlet and Neumann boundary data. We would further like to draw the attention of the reader to the fact that in \cite{Harr09}, the problem of determination of $a(x)$ and $b(x)$ from such Dirichlet and Neumann boundary data pairs is studied in the context of the inverse problem arising in DOT, where the coefficients $a(x)$ and $b(x)$ are interpreted as diffusion and absorption parameters respectively and we will continue to use this terminology for the parameters of interest throughout this article. In light of this result, in this work, \textit{we propose a Bayesian level-set method for the simultaneous recovery of such piecewise-constant parameters of the PDE from noisy boundary data. Furthermore, we show a well-posedness result for this inverse problem in a Bayesian setting.} To the best of our knowledge, a Bayesian approach to this inverse problem which is inspired by DOT, has not been done before.

\subsection{Non-Parametric Bayesian Approach}

\par Now we will give a review of the non-parametric Bayesian approach to inverse problems that we have used here; and which was formulated by Stuart in \cite{stuart_10} and applied for a number of inverse problems such as in  \cite{dashti10,dashti_13,dashti_11,stuart_16,stuart16b,stuart16a,Nickl20}. To illustrate the main ideas, consider the following abstract set-up. Let $X$ and $Y$ be two separable Banach spaces with their associated Borel $\sigma-$ algebras and consider a measurable map $G:X\to Y$. Suppose we want to reconstruct the parameter $v\in X$ given the measured (nosiy) data $y\in Y$ where:
\begin{align}\label{eq:1.3}
  y=G(v)+\eta  
\end{align} 
and $\eta\in Y$ represents noise. In the Bayesian approach, we model $(v,y)\in X \times Y$ as a random variable. The solution to the inverse problem is then the conditional distribution of the random variable, $v|y$, which is also referred to as the posterior distribution. If the space $X$ is finite-dimensional (say $\mathbb{R}^n$), we can define the solution equivalently as a posterior (Lebesgue) density. To get this posterior density, one first defines a prior density $\rho_0(v)$ on the space of parameters. The prior captures our subjective belief about the parameter before observing any data, and can also be thought of as a regularization tool, see e.g. \cite{dashti17}. Subsequently, we use some form of the Bayes' theorem to update our (posterior) belief about the parameter after we observe the data. Assuming that the noise has a density $\rho$, we can write the following version of Bayes' theorem in terms of densities,
\begin{align*}
    \underbrace{\rho(v|y)}_{posterior}= \frac{1}{Z_y} \underbrace{\rho_0(v)}_{prior}\times \underbrace{\rho(y-G(v))}_{likelihood} \text{ where } Z_y=\int \rho_0(v)\rho(y-G(v))dv>0.
\end{align*}
Here $\rho(v|y)$ is the posterior density and the term $Z_y$ is a normalization term. These densities are Radon-Nikodym derivatives of the respective measures (prior or posterior) with respect to the Lebesgue measure on $\mathbb{R}^n$. Now, consider the case, when $X$ is an infinite-dimensional space, e.g. a Hilbert space of functions. There is no analogue of a translation invariant Lebesgue measure on such infinite-dimensional function spaces. This means that the Bayes' theorem needs to be formulated in a (Lebesgue) density-free manner. Next we will briefly describe the ideas that lead to a measure-theoretic formulation of the Bayes' theorem; these ideas have been presented in great detail in \cite{ dashti17, stuart_10}. We begin by placing a prior probability measure on the space of parameters, i.e. we assume $v\sim \mu_0$ where $\mu_0$ is some probability measure on the space $X$. Let the noise be independent of $v$ and be distributed according to some measure, $\eta \sim \mathbb{Q}_{0}$. Assuming that the data is given according to the additive noise model \eqref{eq:1.3}, we can say that the random variable $y|v$ is distributed according to the measure $\mathbb{Q}_v$, which is a translate of $\mathbb{Q}_0$. For instance, if $\mathbb{Q}_0$ is a Gaussian measure with mean $0$, then the measure $\mathbb{Q}_v$ is a translate of $\mathbb{Q}_0$ with mean $G(v)$. Furthermore, if $\mathbb{Q}_v$ is absolutely continuous with respect to the measure $\mathbb{Q}_0$, i.e. $\mathbb{Q}_{v}\ll\mathbb{Q}_{0}$, then there exists a positive Radon-Nikdoym density given by, $\frac{d \mathbb{Q}_v}{d\mathbb{Q}_0} (y)=\exp{(-\Phi(v,y))} $, where,  $\Phi:X\times Y\to \mathbb{R}$ is the `log-likelihood' function, which is also sometimes referred to as a `potential.' Now we consider the following two product measures, $\nu_0=\mu_0\times \mathbb{Q}_0$ and $\nu=\mu_0\times \mathbb{Q}_v$ on the product space $X\times Y$. Then, we have the following analogue of the Bayes' theorem on infinite-dimensional spaces:
\begin{theorem}{\cite[Theorem 14.]{dashti17}} \label{th:1.1}
 Let $\Phi:X\times Y\to \mathbb{R}$ be $\nu_{0}$ measurable and let $Z_y$ defined as $\int_{X}\exp(-\Phi(v,y)) d\mu_{0}:=Z_y>0$ for $\mathbb{Q}_0$ a.s. $y$, then the conditional distribution of $v|y$ denoted by $\mu^y$ exists under $\nu$. Furthermore, $\mu^y\ll \mu_0$ and \begin{align}
\frac{d\mu^y}{d\mu_0}(u)=\frac{1}{Z_y}\exp(-\Phi(v,y)).
\end{align}
\end{theorem}
Following the Bayesian philosophy, in which the solution to an inverse problem is to be thought of as a posterior measure, Theorem \ref{th:1.1} should then be interpreted as a statement about the existence of a solution to an inverse problem. In addition, provided that the potential, $\Phi(v,y)$, meets certain regularity conditions, we can show that the Bayesian inverse problem is, in fact, well-posed, see \cite[Section 4]{dashti17} and Theorem \ref{thm:5.2} below. Given that inverse problems are ill-posed and sometimes, severely so, the well-posedness in the Bayesian formulation of inverse problems may seem surprising. However, if we note that the prior introduces a form of regularization in the problem, this well-posedness in the Bayesian setting can be more easily understood. {We would like to note here, that while classical regularization methods are ultimately faster than Bayesian methods when applied to inverse problems in which the parameter of interest lies in a very high dimensional space, \textit{one major advantage of the Bayesian paradigm is that it gives us a way to quantify uncertainty in the predicted value of the parameter, see our numerical results below.}} In this article, we follow the approach outlined above to show that in a Bayesian framework; the simultaneous reconstruction of both the absorption and diffusion parameters can be done in a way that is robust to perturbations in data, see Theorem \ref{thm:5.2}. Subsequently, we also provide numerical reconstructions in support of our theoretical results.

\subsection{Related Research}

\par We will now give a partial survey of some important results for coefficient inverse problems (CIP) in general, and statistical CIP in particular. The DOT inspired CIP as described in this article, is modeled by an elliptic differential equation which has two coefficients.  A related CIP is that of coefficient (conductivity of the medium) determination in Calder\'{o}n problem from measured boundary data. Mathematically, this corresponds to the case when $b(x)=0$ in \eqref{eq:1.1}. The Calder\'{o}n problem is used to model the imaging modality of electrical impedance tomography (EIT) and has been studied by many authors since it was first described in \cite{cald_80}. Unique recovery of conductivity in the Calder\'{o}n problem from boundary Dirichlet and Neumann data was established in \cite{nach,SU} and  stability estimates for the reconstruction were given in \cite{ales,NS10}. For a survey of results on the Calder\'{o}n problem, we point the reader to \cite{Salo_lec,U_sur} and many references therein. In a foundational 1981 work, Bukhgeim and Klibanov introduced the use of Carleman estimates to prove global uniqueness theorems for multidimensional coefficient inverse problems (MCIPs) with non-overdetermined data, \cite{BK_81}. It has also inspired effective globally convergent numerical methods (GCMs), including the convexification methods, that were developed to ensure stable recovery of coefficients without prior knowledge of the solution neighborhood and have been used in important physical examples, e.g. \cite{KT_03,K_19}. \cite{arr_99,arr_scho} are excellent surveys for the DOT problem. Authors in \cite{arr_99,arr_scho,khan_05,natt_wub} show that the DOT inverse problem is exponentially ill-posed and, in general, reconstruction procedures in DOT are very unstable. However, in \cite{Harr09}, the author showed that simultaneous recovery of both absorption and diffusion is, in principle, possible if we consider the case when $a(x)$ is piecewise constant and $b(x)$ is piecewise analytic. A simplified one-parameter DOT model was studied in a deterministic setting  in \cite{Gaburro16} and in a stochastic setting in \cite{Abhi_22b}. {In this article, we extend our study to simultaneous determination of both the absorption and diffusion-like coefficients, which to our knowledge, has not been done in a Bayesian setting before.} The field of statistical inverse problems, though relatively new, has become increasingly important in the past couple of decades due to computational advances that have made statistical algorithms for reconstruction feasible. In the EIT setting, statistical inversion procedures for recovery of conductivity was given in \cite{hank_11,jin_12,kaipio_00, silt_13}. In \cite{khan_strauss}, a statistical inversion procedure for EIT using mixed total variation and non-convex regularization prior was given. The aforementioned statistical approaches to inverse problems fall under the umbrella of what are called parametric Bayesian methods. The non-parametric Bayesian point of view for inverse problems taken in the present article was first proposed by Stuart in \cite{stuart_10}. This Bayesian perspective for coefficient and source determination inverse problems arising from PDEs was then taken forward in a series of articles \cite{NiAb19,dashti10,dashti_13,dashti_11,Nickl_matt,newton_20,nickl_von1,Nickl_sohl}. These methods are distinguished by their use of Gaussian process priors and in \cite{Sai_22}, the authors provide efficient algorithms for Bayesian inverse problems with such priors. Furthermore, applicability of Bayesian level-set procedures for determination of piecewise smooth or piecewise constant parameters from boundary, or, internal measurements was given in \cite{stuart_16,stuart16a,stuart16b}. We would like to note here the important contribution of \cite{osher05,dorn06,dorn00,sant96} wherein level-set methods were deployed for certain inverse problems in a deterministic setting. As a survey of the above results indicate, most of the efforts so far in statistical CIP have been focused on the problem of EIT in which recovery of just one unknown parameter is sought. For the CIP that arise in inverse problems such as DOT, we are interested in simultaneous recovery of two unknown parameters. Bayesian formulation of inverse problems can be thought of as being in the same spirit as the formidable and highly efficient globally convergent methods such as the ones using convexification approach, \cite{KT_03, K_19}. This is because the regularization imposed by the prior in Bayesian formulation and the subsequent use of MCMC to explore the posterior has the effect of preventing the algorithms from being stuck in local minima. In order to show the applicability of the Bayesian methods to this particular inverse problem, we first theoretically prove a well-posedness result for the Bayesian inverse problem using some newly derived regularity result for the likelihood. After establishing the well-posedness result, we outline the reconstruction method and provide extensive numerical simulations to show the efficacy of the proposed method in not just giving us a reconstruction but also a way to quantify uncertainty in the predictions.


\subsection{Organization}
\par The organization of this article is as follows. In section \ref{sec2} we present the mathematical formulation of the inverse problem and derive some preliminary results in section \ref{sec3} that are used later in the article. In section \ref{prior}, we illustrate how the level-set prior is modeled and prove a well-posedness result for the Bayesian inverse problem in section \ref{sec5}. We present the numerical simulation results to demonstrate the efficacy of our proposed algorithms in Section 6. In Section 7, we summarize the work done in this article.
\section{Mathematical formulation of the problem} \label{sec2}
We consider the following semi-discrete mathematical formulation of the problem, similar to the semi-discrete complete electrode model introduced for EIT \cite{SoChIs_92}. Let $\Omega\in \mathbb{R}^2$ be the region to be imaged with its boundary denoted by $\partial \Omega$. Assume that the measurements can only be made on $L$ discrete boundary elements, $\{O_l \}_{l=1}^L$, placed along $\partial \Omega$. We assume that the parameters $a(x)$, $b(x)$ are real valued, positive, bounded functions in $L^{\infty}(\bar{\Omega})$. In analogy with the complete electrode model in EIT \cite [section 3]{SoChIs_92}, we consider the following semi-discrete formulation of the proposed model (see also, \cite [eqns. (7.27-28)]{natt_wub}):
\begin{align} \label{eq:1}
-\nabla\cdot (a(x)(\nabla u(x)))+b(x) u(x) &=0, \quad x\in \Omega\\
u(x)+2 a\frac{\partial u}{\partial \nu} &= U_l, \quad x\in O_l, \quad l\in \{1,\dots, L\}\label{eq:2.2}\\
a\frac{\partial u}{\partial \nu}&=0 \quad \text{on } \partial \Omega\setminus \bigcup_{l=1}^{L}O_l\label{eq:2.3}.
\end{align}
We also assume that the average value of the prescribed Neumann data on each boundary element denoted by $\bar{F}_{l}=\frac{1}{\lvert O_l \rvert}\int_{O_l}- a\frac{\partial u}{\partial \nu} dS$ is known. Here $dS$ is the surface measure on $\partial \Omega$ and $\lvert O_l\rvert$ is the surface measure of the $l^{th}$ boundary element $O_l$. Since $\lvert O_l\rvert$ is known, hence we assume that the flux, $F_l:=\bar{F}_l\times \lvert O_l \rvert$ satisfies
\begin{align} \label{eq:2}
\int _{O_l} -a\dfrac{\partial u}{\partial \nu}dS = F_l
\end{align}
 In what follows, we will suppress the explicit dependence of the functions $a(x)$ and $b(x)$ on $x$ whenever it can be easily understood form the context.
Formally, for each $j\in \{1,\dots,J\}$, let us represent the true discrete Robin-data, $U^{(j)}$, corresponding to the applied flux $F^{(j)}$ on the boundary elements by:
\begin{align}
U^{(j)}=\mathcal{G}_j (a,b)
\end{align}
Consider the case, when the measured Robin data is corrupted by a Gaussian noise. Let $y_j$ denote these noisy measurements:
\begin{equation}\label{eq:2.6}
y_j=\mathcal{G}_{j}(a,b)+\eta _j; \quad j\in \{1,\dots,J\}\quad \text{and }\eta_j\sim N(0,\Gamma_0) \ \text{i.i.d.}
\end{equation}
Here $N(0,\Gamma_0)$ is used to denote a Gaussian random variable with mean ${0}$ and variance $\Gamma_0$. Concatenating all the vectors $y_j\in \mathbb{R}^{L}$ we can write:
\begin{align}\label{eq:2.7}
y=\mathcal{G}(a,b)+\eta
\end{align}
where $y\in \mathbb{R}^{LJ}$ and $\eta \sim N(\textbf{0},\Gamma)$ where $\rm{\Gamma}=diag(\Gamma_0,\dots,\Gamma_0)$. The statistical inverse problem can now be formulated as recovery of the parameters, $a$ and $b$, from observed (noisy) data $y$. We would like to remind the reader about the severe ill-posedness of the problem as stated. On the one hand, we have finite-dimensional (noisy) data from the measurements, and yet, we seek to reconstruct two functions that, in principle, live in some infinite-dimensional function space. The sensitivity of the reconstructions to the priors, can thus, be well-understood and we need to choose a prior measure that is well-suited for a given experimental set-up. In the present case, when we consider the functions $a(x)$ and $b(x)$ to be piecewise constant functions; geometric priors, such as the level-set priors, seem to be a natural choice. The modeling of such priors will be discussed in section \ref{prior}; for now we derive some preliminary results in section \ref{sec3}. These preliminary results are needed in order to extend the theory proposed for Bayesian level-set method for EIT \cite{stuart_16} to the simultaneous estimation of two-parameters as in the inverse problem for DOT.
\section{Preliminary Results}\label{sec3}
We consider first a weak formulation of the model given by equations \ref{eq:1}- \ref{eq:2}. For this we introduce the space of solutions $\mathbb{H}=H^1(\Omega)\bigoplus \mathbb{R}^L$. A pair $(v,V)\in \mathbb{H}$ is such that $v\in H^1(\Omega)$ and $V\in \mathbb{R}^L$. Furthermore, $\lVert(v,V)\rVert_{\mathbb{H}}=\lVert v\rVert_{H^1}+\lVert V\rVert_{l^2}$. We have the following proposition which is an analogue of \cite [Proposition 3.1]{SoChIs_92}:
\begin{prop}\label{prop:1}
Let $\Omega, a(x), b(x)$ be as above. Assume that $(u,U)=(u(x),(U_l)_{l=1}^L)\in \Hb$ and $u\in H^1(\Omega)$ is a weak solution to:
\begin{align*}
    -\nabla \cdot (a\nabla u)+b u=0
\end{align*}
subject to the boundary conditions:
\begin{align*}
   u(x)+2 a(x)\frac{\partial u}{\partial \nu} &= U_l, \quad x\in O_l, \quad l\in \{1,\dots, L\}\\
a\frac{\partial u}{\partial \nu}&=0 \quad \text{on } \partial \Omega\setminus \bigcup_{l=1}^{L}O_l. 
\end{align*} We also assume that:
\begin{align*}
    \int _{O_l} -a\dfrac{\partial u}{\partial \nu}dS = F_l
\end{align*}
Consider the bilinear form:
\begin{align}\label{eq:3.1}
    B: \ & \Hb \times \Hb \to \Rb\nonumber \\
    B((u,U),(w,W))=\int_{\Omega}a \nabla u \cdot \nabla w dx+ &\int_{\Omega}b u w dx +\sum_{l=1}^L\frac{1}{2}\int_{O_l}(u-U_l)(w-W_l) dS
\end{align}
Then for any $(v,V)\in \Hb$, we have:
\begin{align}\label{eq:3.2}
B((u,U),(v,V))=-\sum_{l=1}^L F_l V_l.
\end{align}
Conversely, if $(u,U)\in \Hb$ satisfies equation \ref{eq:3.2} for any $(v,V)\in \Hb$, then it also satisfies equations \ref{eq:1}-\ref{eq:2}. 
\end{prop}
\begin{remark}\label{rm:3.1}
    At times, in order to emphasize the parameters $a(x)$ and $b(x)$, we will use the following notation for the bilinear form: $B((u,U),(v,V)):=B((u,U),(v,V); a,b)$
\end{remark}
\begin{proof}
The proof here is similar to that of \cite [Proposition 3.1]{SoChIs_92}, except for the appearance of an additional term, namely, $\int _{\Omega}b u w dx$ in the bilinear form. We reproduce the arguments here for the sake of completeness. First of all, note that we can express $a\frac{\partial u}{\partial \nu}$ in a more succinct form as:
\begin{align}\label{eq:3.3}
a\frac{\partial u}{\partial \nu}=\sum_{l=1}^L \frac{1}{2}(U_l-u)\chi_l 
\end{align}
where $\chi_l$ is the characteristic function of the boundary element given by $O_l$. For any $v\in H^1(\Omega)$, we have from equation \eqref{eq:1},
\begin{align} \label{eq:3.4}
\int _{\Omega} a \nabla u \cdot \nabla v dx+ \int_{\Omega} b u v dx -\int_{\partial \Omega} a \frac{\partial u}{\partial \nu} v dS =0.
\end{align}
Substituting \eqref{eq:3.3} in \eqref{eq:3.4}, we get:
\begin{align}\label{eq:3.5}
\int _{\Omega} a \nabla u \cdot \nabla v dx+ \int_{\Omega} b u v dx + \sum_{l=1}^L \frac{1}{2}\int_{O_l}(u-U_l) v dS=0
\end{align} 
On the other hand, \begin{align}\label{eq:3.6}
   \int_{O_l}(u-U_l) dS=\int_{O_l}-2 a \frac{\partial u}{\partial \nu} dS &=  2 F_l \nonumber \\
   \implies \sum _{l=1}^L \frac{1}{2} \big( 2 F_l V_l-\int_{O_l}(u-U_l) V_l dS\big)&=0
\end{align}
where $V=(V_l)_{l=1}^L\in \Rb^L$.
Adding \eqref{eq:3.5} and \eqref{eq:3.6}, we get: 
\begin{align}
\int _{\Omega} a \nabla u \cdot \nabla v dx+ \int_{\Omega} b u v dx + \sum_{l=1}^L \frac{1}{2}\int_{O_l}(u-U_l) (v-V_l) dS =-\sum_{l=1}^L F_l V_l.
\end{align}
This establishes the first part of the proposition. Now to show the converse, let us assume that $(u,U)\in \Hb$ satisfies equation \eqref{eq:3.2} for all $(v,V)\in \Hb$. In particular, choose any $v\in C^{\infty}(\Omega)$ supported in $\Omega$ and $V=(V_l)_{l=1}^{L}=\mathbf{0}\in \Rb^L$. Thus, from equation \eqref{eq:3.2} we have:
\begin{align}
\int_{\Omega} a \nabla u \cdot \nabla v dx+\int_{\Omega} b u v dx =0
\end{align}
Since $u\in H^1(\Omega)$, from this weak form we actually conclude that eqn. \eqref{eq:1} holds. Now consider an arbitrary $v\in H^{1}(\Omega)$ and $V=\mathbf{0}$ as before. Since we have already established eqn. \eqref{eq:1} to hold for $x\in \Omega$, thus we have that eqn. \eqref{eq:3.4} holds. On the other hand, by assumption \eqref{eq:3.2} is satisfied. Thus, for any $v\in H^1(\Omega)$ and for $V=\mathbf{0}$ we have:
\begin{align}\label{eq:3.9}
\int _{\Omega} a \nabla u \cdot \nabla v dx+ \int_{\Omega} b u v dx + \sum_{l=1}^L \frac{1}{2}\int_{\partial \Omega}(u-U_l) \chi_l(v) dS =0.
\end{align}
Thus on comparing \eqref{eq:3.4} and \eqref{eq:3.9}, we get for any $v\in H^1(\Omega)$:
\begin{align}\label{eq:3.10}
 \sum_{l=1}^L \frac{1}{2}\int_{\partial \Omega}(u-U_l) \chi_l v dS = -\int_{\partial \Omega}   a\frac{\partial u}{\partial \nu} v dS.
\end{align}
It is now easy to see that this implies that \eqref{eq:2.2} and \eqref{eq:2.3} hold as well. Finally, we take an arbitrary $v\in H^{1}(\Omega)$ and $V\in \Rb^L$, then from \eqref{eq:3.9}, \eqref{eq:3.10} and \eqref{eq:3.2} we get:
\begin{align}\label{eq:3.11}
    \sum_{l=1}^L \frac{1}{2}\int_{\partial \Omega}(u-U_l) \chi_l V_l dS&=-\sum_{l=1}^L F_lV_l\nonumber\\
    \implies \int_{\partial \Omega} a\frac{\partial u}{\partial \nu} dS&= F_l, \quad l\in \{1,\dots,L\}.
\end{align}
Thus \eqref{eq:2} holds as well. We also draw the attention of the reader to the fact that the integral on the left hand side of the equation \eqref{eq:3.11} is defined in a classical sense, see also remark after the statement of \cite [Proposition 3.1]{SoChIs_92}.
\end{proof}
Now that we have a weak formulation of the discrete formulation of the CIP considered here, we want to know if this weak formulation admits a unique solution. To that end, we will check that the conditions of the Lax-Miligram lemma hold. Consider the bilinear form introduced in equation \eqref{eq:3.1}. Indeed, the bilinearity is easy to check and we skip it here. We will now show that the bilinear form is bounded and coercive. More particularly, we have the following lemma:
\begin{lemma}
Consider the bilinear form $B((u,U),(v,V))$ introduced in Proposition \ref{prop:1} above. For $(u,U)\in \Hb$ and $(v,V)\in \Hb$, for some $\gamma$ and $\delta$ independent of $(u,U)$ and $(v,V)$, we have:
\begin{align}
\lvert B((u,U),(v,V))\rvert &\leq \gamma \lVert (u,U)\rVert_{\Hb} \lVert (v,V)\rVert_{\Hb}\\
\lvert B((u,U),(u,U))\rvert &\geq \delta \lVert (u,U)\rVert^2_{\Hb}. 
\end{align}
\end{lemma}
\begin{remark} To avoid introducing new constants we will use the notation L.H.S. $\lesssim$ R.H.S. to mean that the inequality is true up to some finite constant. 
\end{remark}
\begin{proof}
To establish the boundedness and coercivity property, it is easier to work in a norm introduced below which is equivalent to $\lVert (\cdot,\cdot)\rVert_{\Hb}$. For any $(u,U)\in \Hb$ consider the following norm:
\begin{align}
\lVert (u,U)\rVert_{*}^2=\lVert u\rVert_{H^1}^2+\sum_{l=1}^L\int_{O_l}\lvert u(x)-U_l\rvert^2 dS.
\end{align}
We will show that the two norms are equivalent by establishing the following inequality:
\begin{align} \label{eq:3.15}
    c\lVert (u,U)\rVert_{*}\leq \lVert (u,U)\rVert_{\Hb}\leq C \lVert (u,U)\rVert_{*}
\end{align}
for some constant $c$ and $C$. First of all, observe that 
\begin{align}\label{eq:3.16}
 \lVert (u,U)\rVert_{*}^2\lesssim \lVert u\rVert_{H^1}^2+ 2\sum_{l=1}^L \int_{O_l} \lvert u(x)\rvert ^2 dS+  2\sum_{l=1}^L \lvert O_l\rvert U_l^2 dS 
\end{align}
where we have used the fact that $(p-q)^2\leq 2 p^2+ 2q^2$.
Now, \begin{align}\label{eq:3.17}
    \sum_{l=1}^L \int_{O_l} \lvert u(x)\rvert ^2 dS\leq  \int_{\partial \Omega} \lvert u(x)\rvert ^2 \lesssim \lVert u\rVert^2_{H^{1/2}(\partial \Omega)}\lesssim \lVert u \rVert_{H^1}^2.
\end{align} 
Substituting \eqref{eq:3.17} in \eqref{eq:3.16} we get, $\lVert (u,U)\rVert_{*}^2\lesssim \lVert (u,U)\rVert_{\Hb}^2$. It remains to show that the right hand inequality in \eqref{eq:3.15} is also true. We will prove this by contradiction. Suppose the inequality, $\lVert (u,U)\rVert_{\Hb}\leq C \lVert (u,U)\rVert_{*}$ does not hold for any $C$. Then we have a sequence, $(u_n,U_n)\in \Hb$ such that $\lVert (u_n,U_n)\rVert_{\Hb}^2\to 1$ with $1-\frac{1}{n^2}\leq \lVert (u_n,U_n)\rVert_{\Hb}^2\leq 1+\frac{1}{n^2}$ while $\lVert (u_n,U_n)\rVert_{*}^2\leq \frac{1}{n^2}.$ This in turn implies that $(u_n)$ is a bounded sequence in $H^1(\Omega)$ over a bounded domain $\Omega$. Thus by the compact embedding of Sobolev spaces over bounded domains, \cite{evans_book}, we know that there exists a subsequence of $(u_n)$ which we still denote by $(u_n)$ which converges to some $u$ in $L^2(\Omega)$. However, recall that we also have $\lVert (u_n,U_n)\rVert_{*}^2\leq \frac{1}{n^2}.$ In particular, $\lVert u_n\rVert^2_{L^2(\Omega)}+\lVert \nabla u_n\rVert^2_{L^2(\Omega)}\leq 1/n^2$. Together, this means that $u_n\to 0$ in $L^2(\Omega)$ (and also in $H^1(\Omega)$).  Furthermore, we also have that for each $l\in \{1,\dots, L\}$:
\begin{align}
\int_{O_l}\lvert u(x)-U_{n_l}\rvert ^2 dS &\leq \frac{1}{n^2} \quad \text{{(from 3.14)}}\nonumber\\
\implies \int_{O_l}\lvert u(x)\rvert ^2 dS- 2U_{n_l}\int_{O_l}u(x) dS + U_{n_l}^2\lvert O_l\rvert  &\leq \frac{1}{n^2}\nonumber\\
\implies U_{n_l}^2\lvert O_l\rvert \leq \frac{1}{n^2}+ 2U_{n_l}\int_{O_l}\lvert u(x)\rvert dS\nonumber \\
\implies  U_{n_l}^2\lvert O_l\rvert \leq \frac{1}{n^2} +\tilde{C}(1+\frac{1}{n^2})\lvert u_n \rvert _{H^1(\Omega)}.
\end{align}
Since $u_n\to 0$ in $H^1(\Omega)$, thus we have that $U_{nl}\to 0$ for each $l$, i.e. $U_n\to \boldsymbol{0}\in \Rb^L$. This would then imply that $\lVert (u_n,U_n)\rVert_{\Hb}\to 0$ which contradicts that our assumption that it converges to 1 instead. This establishes the equivalence of the two norms as claimed. 
\par \noindent Coming back to the proof of boundedness and coercivity of the bilinear form, we will use the fact that $0<m\leq a,b \leq M, M>1$. Thus:
\begin{align}
 \lvert B((u,U),(v,V))\rvert&=\lvert\int_{\Omega}a \nabla u \cdot \nabla v dx+ \int_{\Omega}b u v dx +\sum_{l=1}^L\frac{1}{2}\int_{O_l}(u-U_l)(v-V_l) dS \rvert \nonumber\\ 
  &\leq M \lvert(\int_{\Omega} \nabla u \cdot \nabla v dx+ \int_{\Omega} u v dx +\sum_{l=1}^L\int_{O_l}(u-U_l)(v-V_l)) dS\rvert\nonumber\\
  &\leq {M} \lVert (u,U)\rVert_{*} \lVert (v,V)\rVert_{*} \quad\text{(Cauchy-Schwartz inequality)}\nonumber\\
  &\lesssim \lVert (u,U)\rVert_{\Hb} \lVert (v,V)\rVert_{\Hb} \quad \text{(norm equivalence)}.
\end{align}
This shows boundedness. For coercivity, 
\begin{align}
\lvert B((u,U),(u,U))\rvert &\geq m_0 (\int_{\Omega} \lvert\nabla u \rvert^2 dx+ \int_{\Omega} u^2 dx +\sum_{l=1}^L\int_{O_l}(u-U_l)^2)\nonumber\\
& \geq m_0 \lVert (u,U)\rVert_{*}^2
\end{align}
and thus, $\lVert (u,U)\rVert_{\Hb}^2\lesssim \lvert  B((u,U),(u,U))\rvert$.
\end{proof} 
Thus, we have shown that the bilinear form is coercive and bounded, and Theorem \ref{th:3.4} follows.
\begin{theorem}\label{th:3.4}
    Let $0<m\leq a,b\leq M <\infty$. Then for any given flux pattern $(F_l)_{l=1}^L$, there exists a unique $(u,U)\in \Hb$ which satisfies \eqref{eq:3.2} for any $(v, V)\in \Hb$.
\end{theorem}
Now we recall the admissible space of parameters which we consider in this article.
\begin{definition}
A space of functions $\mathcal{A}(\Omega)$ will be said to be an admissible space of parameters, if for any $f\in \mathcal{A}(\Omega)$ where $f:\Omega\to \mathbb{R}$  the following holds true:
\begin{itemize}
    \item There exists a $N\in \mathbb{N}$ such that $\bar{\Omega}= \cup_{i=1}^N \bar{\Omega}_i$ where $\{\Omega_i\}_{i=1}^N$ are disjoint open subsets of $\Omega$.
\item $f|_{\Omega_i}=c_i$ where $c_i$ are constants.
\item There exist $c_m$ and $c_M$ such that, $0<c_m<f(x)<c_M$ for all $x\in \Omega$.
\end{itemize}
\end{definition}
From here on, we will assume that $a(x)$ and $b(x)$ are both in $\mathcal{A}(\Omega)$. Consider the solution operator $\mathcal{M}$ for the weak formulation posed by \eqref{eq:3.2}, see also Remark \ref{rm:3.1} .
\begin{align}
\mathcal{M}:&\mathcal{A}(\Omega)\times \mathcal{A}(\Omega)\to \Hb\\
&(a,b)\mapsto (u,U).\nonumber
\end{align}
We will show that the map $\mathcal{M}$ is continuous in the following sense:
\begin{theorem}\label{thm:3.6}
For a fixed flux pattern, $(F_l)_{l=1}^L$, let the map $\mathcal{M}:\mathcal{A}(\Omega)\times \mathcal{A}(\Omega)\to \Hb$ be as above. Consider $a(x)$ and $b(x)$ to belong to $\mathcal{A}(\Omega)$ and let $(a_{\epsilon}(x))$ and $ (b_{\epsilon}(x))$ be sequences in $\mathcal{A}(\Omega)$ which converge in measure to $a(x)$ and $b(x)$ respectively. Then $\lVert \mathcal{M}(a_{\epsilon},b_{\epsilon}) -\mathcal{M}(a,b)\rVert_{*}\to 0.$ Furthermore, if $\Pi:\Hb\to\mathbb{R}^l$ is the projection map, such that, $\Pi(v,V)=V$, then we also have the continuity of the composition map $\Pi\circ \mathcal{M}$. More precisely, $\lVert \Pi\circ \mathcal{M}_{\epsilon}-\Pi\circ \mathcal{M}\rVert_{l^2}\to 0.$
\end{theorem}
\begin{proof}

This proof adapts the proof of \cite[Proposition 2.7]{stuart_16} that was done for the problem of electrical impedance tomography in which we have a single unknown parameter (conductivity) to reconstruct. Introducing the notation, $v^{\prime}_{\epsilon}:=(v_{\epsilon},V^{\epsilon})=\mathcal{M}(a_{\epsilon},b_{\epsilon})$ and $v^{\prime}:=(v,V)=\mathcal{M}(a,b)$, we have for all $w^{\prime}:=(w,W)\in \Hb$:
$$B(v^{\prime}_{\epsilon},w^{\prime};a_{\epsilon},b_{\epsilon})=\sum_{l=1}^L F_l W_L=B(v^{\prime},w^{\prime}; a,b).$$
Thus, \begin{align}\label{eq:3.22}
B(v^{\prime}_{\epsilon},w^{\prime}; a_{\epsilon},b_{\epsilon})&-B(v^{\prime},w^{\prime};a_{\epsilon},b_{\epsilon})+B(v^{\prime},w^{\prime};a_{\epsilon},b_{\epsilon})-B(v^{\prime},w^{\prime};a,b)=0\nonumber\\
\int_{\Omega}a_{\epsilon}(\nabla v_{\epsilon}-\nabla v)\cdot \nabla w dx&+\int_{\Omega}b_{\epsilon}(v_{\epsilon}- v)\cdot w dx+\int_{\Omega}(a_{\epsilon}-a)\nabla v\cdot \nabla w dx+\int_{\Omega}(b_{\epsilon}-b) v\cdot  w dx\nonumber\\
&+\sum_{l=1}^L \frac{1}{2}\int_{O_l}(w-W_l)[(v_\epsilon-v)-(V^{\epsilon}_l-V_l)]dS=0.
\end{align}
If we substitute, $w^{\prime}=(v_{\epsilon}-v, V^{\epsilon}-V)$ in \eqref{eq:3.22}, we get:
\begin{align} \label{eq:3.23}
\int_{\Omega}a_{\epsilon}(\nabla v_{\epsilon}-\nabla v)^2 dx&+\int_{\Omega}b_{\epsilon}(v_{\epsilon}- v)^2 dx
+\sum_{l=1}^L \frac{1}{2}\int_{O_l}[(v_\epsilon-v)-(V^{\epsilon}_l-V_l)]^2 dS \nonumber \\ &\leq \int_{\Omega}\lvert (a_{\epsilon}-a)\rvert \lvert \nabla v\cdot \nabla (v_{\epsilon}-v)\rvert  dx+\int_{\Omega}\lvert(b_{\epsilon}-b)\rvert \lvert v\cdot  (v_{\epsilon}-v)\rvert dx
\end{align}
Using the fact that $a_{\epsilon}, b_{\epsilon},a$ and $b$ are all bounded above and below, we can rewrite \eqref{eq:3.23} as:
\begin{align}\label{eq:3.24}
\lVert v_{\epsilon}^{\prime}-v^{\prime}\rVert_{*}^2&\lesssim \bigg(\int_{\Omega} \lvert (a_{\epsilon}-a) \nabla v\rvert^2 dx\bigg)^{\frac{1}{2}}\lVert \nabla(v_{\epsilon}^{\prime}-v^{\prime})\rVert_{2}+\bigg(\int_{\Omega}\lvert(b_{\epsilon}-b) v\rvert^2 dx\bigg)^{\frac{1}{2}} \lVert (v_{\epsilon}^{\prime}-v^{\prime})\rVert_{2}\nonumber\\
&\lesssim \bigg[\bigg(\int_{\Omega} \lvert (a_{\epsilon}-a)\rvert^2 \lvert \nabla v\rvert^2 dx\bigg)^{\frac{1}{2}}+\bigg(\int_{\Omega}\lvert(b_{\epsilon}-b)\rvert^2 \lvert v\rvert^2dx\bigg)^{\frac{1}{2}} \bigg]\lVert (v_{\epsilon}^{\prime}-v^{\prime})\rVert_{*}\nonumber\\
\implies \lVert v_{\epsilon}^{\prime}-v^{\prime}\rVert_{*}&\lesssim \bigg[\bigg(\int_{\Omega} \lvert (a_{\epsilon}-a)\rvert^2 \lvert \nabla v\rvert^2 dx\bigg)^{\frac{1}{2}}+\bigg(\int_{\Omega}\lvert(b_{\epsilon}-b)\rvert^2 \lvert v\rvert^2 dx\bigg)^{\frac{1}{2}} \bigg]
\end{align}
Recall that $a_{\epsilon}\to a$ and $b_{\epsilon}\to b$ in measure. Moreover $(a_{\epsilon}), (b_{\epsilon})$ are bounded sequences and the measure of the bounded domain $\Omega$ is finite. Thus each of the integrands on the right hand side of \eqref{eq:3.24} converges to $0$ in measure, see e.g. \cite [Proposition 2.7]{stuart_16}. From this we conclude,  $\lVert v_{\epsilon}^{\prime}-v^{\prime}\rVert_{*}= \lVert \mathcal{M}(a_{\epsilon},b_{\epsilon}) -\mathcal{M}(a,b)\rVert_{*}\to 0.$ The continuity of the composition map follows in exactly the same manner as in the proof of \cite [corollary 2.8]{stuart_16}.
\end{proof}
\begin{remark}\label{rm:3.7}: The operator $\Mca:\Ac (\Omega)\times \Ac(\Omega)\to \Hb$ will be called a measurement operator. More particularly, we will denote by $\Mca_j:\Ac(\Omega) \times \Ac(\Omega)\to\Hb$, the $j-$th measurement operator corresponding to a given flux pattern $F^{(j)}\in \Rb^L.$
\end{remark}
\section{Modeling a level-set prior}\label{prior}
Diffusive and absorptive targets will be expressed with the help of level-set functions, one for each. To make ideas precise, consider the case where diffusive and absorptive coefficients are expressed as:
\begin{align}
    a(x)=\sum_{i=1}^M a_{i}\mathbb{I}(\Omega^d_{i}), \quad \quad b(x)=\sum_{\ip=1}^N b_{\ip} \mathbb{I}(\Omega^a_{\ip})
\end{align} for some $M,N\in \Nb$. Here, $\mathbb{I}(S)$ denotes the characteristic function of some set $S\subset \Omega$. Also, for $i\neq k$, and $\ip \neq k^{\prime}$, $\Omega^d_i \cap \Omega^d_k=\emptyset,\ \Omega^a_{\ip}\cap \Omega^a_{k^{\prime}}=\emptyset$. Besides, $\cup_{i=1}^M \Omega^d_i=\Omega=\cup_{\ip=1}^N\Omega^a_{\ip}$. The constants $a_i $ and $ b_{\ip}$ are bounded, strictly positive numbers such that $a(x)\in \mathcal{A}(\Omega)$ and $b(x)\in \mathcal{A}(\Omega)$. From here on, we assume that the constants $\{a_{i}\}_{i=1}^M$ and $\{b_{\ip}\}_{\ip=1}^N$ are known to us and what we need to determine are the regions $\{\Omega_{i}^d\}_{i=1}^M$ and $\{\Omega_{\ip}^a\}_{\ip=1}^N$ from certain boundary measurements. Choose numbers $\{c_i\}_{i=1}^M$ and $\{d_{\ip}\}_{\ip=1}^N$ with $c_0<c_1\dots<c_M$ and $d_0<d_1<\dots<d_N$, and  two continuous functions, $u_{i}:\Omega \to \Rb$ for $i=1,2$, such that:
\begin{align}
\Omega^d_i&=\{x\in \Omega: c_{i-1}\leq u_1(x)<c_i\} \quad \text{and}\\
\Omega^a_{\ip}&=\{x\in \Omega:d_{\ip-1}\leq u_2(x)<d_{\ip}\}.
\end{align}
In effect, the boundaries of the diffusive and absorptive regions are demarcated by the level-sets of the continuous functions $u_1$ and $u_2$ respectively, where the corresponding diffusive and absorptive level-sets are given by:
\begin{align}L^d_i&=\{x\in \Omega: u_1(x)=c_i \} \quad\text{and}\\
L^a_{\ip}&=\{x\in \Omega: u_2(x)=d_{\ip}\}.
\end{align} 
Let $H_1$ and $H_2$ be two operators such that $H_{l}:C(\bar{\Omega})\to \mathcal{A}(\Omega)$, $l=1,2$:
\begin{align}
H_1(u_1)=\sum_{i=1}^M a_{i}\mathbb{I}(\Omega^d_{i})=a(x) \quad\text {  and  } \quad H_2(u_2)= \sum_{\ip=1}^N b_{\ip}\mathbb{I} (\Omega^a_{\ip})=b(x).
\end{align}
If the level-sets $L_i^d$ and $L_j^a$ of the functions $u_1$ and $u_2$ respectively each have zero measure then the functions $H_1$ and $H_2$ satisfy a certain continuity in measure property. More precisely we have the following result from \cite {stuart_16}:
\begin{prop}{\cite[Proposition 3.5]{stuart_16}.} \label{prop4.1}
Let $u_1,u_2,H_1$ and $H_2$ be as mentioned above in section \ref{prior}. Assume that $\mathrm{leb}(L_i^d)=0$ for each $i\in\{1,\dots,M\}$ and $\mathrm{leb}(L_j^a)=0$ for each $j\in \{1,\dots,N\}$ where $\mathrm{leb}(\cdot)$ denotes the Lebesgue measure of the sets. Suppose that $(u_{1_{\epsilon}})$ and $(u_{2_{\epsilon}})$ are approximating sequence of functions such that $\lVert u_{k_{\epsilon}}-u_i\rVert_{\infty}\to 0$ for $k=1,2$, then $H_1(u_{1_{\epsilon}})$ and $H_2(u_{2_{\epsilon}})$ converge respectively to $H_1(u_1)$ and $H_2(u_2)$ in measure.
\end{prop}
Now we can reformulate the regression framework for the observations given by \eqref{eq:2.6} in terms of operators $\Pi$, $\Mca_j$, $H_1$ and $H_2$ as:
\begin{align}\label{eq:4.6}
y_j=\Pi\circ \Mca_j (H_1(u_1),H_2(u_2))+\eta_j, \quad j\in \{1,\dots,J\}\quad \text{and }\eta_j\sim N(0,\Gamma_0).
\end{align}
As mentioned in remark \ref{rm:3.7}, $\Mca_j$ is the $j-$th measurement operator corresponding to the flux pattern $F^{(j)}$ for $j\in \{1,\dots,J\}$. Thus we have replaced $\Gc_j(a,b)$ from \eqref{eq:2.6} with $\Pi\circ \Mca_j (H_1(u_1),H_2(u_2))$ and in doing so reformulated the inverse problem as one where the goal is to determine the level-set functions $u_1$, $u_2$ from noisy data. In the Bayesian setting, we think of the coefficients $a(x)$ and $b(x)$ as random variables, which in terms of the formulation given by \eqref{eq:4.6} implies that the functions $u_i(x), i=1,2$ ought to be modeled as random variables. We will also make the assumption that the parameters $a(x),b(x)$ and noise $\eta_j$ are independent, see e.g. \cite[section 2]{mozu21}. Thus we will choose independent prior measures for $u_1$ and $u_2$ such that the level-sets for any sample has measure zero almost surely. This can be ensured by choosing non-degenerate Gaussian measures; as was shown in \cite[Proposition 2.8]{stuart16a}. From proposition \ref{prop4.1}, it will mean that the maps $H_1$ and $H_2$ are almost surely continuous under the chosen priors and will ensure the well-posedness for the Bayesian inverse problem (in particular, the infinite-dimensional Bayes' theorem will hold), see \cite[section 2]{stuart16a}. The prior measures for $u_1$ and $u_2$ are chosen to be $\mu_1^0$ and $\mu_2^0$ where {{$\mu_1^0= N(m_1^0,C_1^0)$ and $\mu_2^0= N(m_2^0,C_2^0)$}} on the space of continuous functions $C(\bar{\Omega})$. Here $m_i^0$ and $C_i^0$ denote the respective mean functions and covariance operators of the Gaussian measures. {The construction of such Gaussian priors has been explained in \cite[section 2.4]{dashti17}. We mention those ideas very briefly here; if $\{\phi^{i}_j\}_{j=1}^{\infty}$ is an eigenbasis and $\{(\gamma^i)^2_{j}\}$ is the corresponding sequence of eigenvalues for each $C_i$, $i=\{1,2\}$ , then any random draw from the measure $\mu^0_i$ is a function of the form $u_i=\sum_j u_i^j \phi^i_j$ where $u_i^j=\gamma^i_j\zeta_j$ and $\{\zeta_j\}$ is a random i.i.d sequence with $\zeta_j\sim N(0,1)$. In the literature, such an eigenbasis $\{\phi^i_j\}$ is called the Karhunen-L\`{o}eve (K-L) basis and the expression for the random function $u_i=\sum_j u_i^j \phi^i_j$ given in form of a convergent series is called a K-L expansion. It should also be noted that the smoothness of the resultant function depends on the nature of the covariance operator. In a lot of applications we choose $C=\mathcal{A}^{-s}$ for some $s>1$ where $\mathcal{A}=\Delta$ is the Laplacian operator. Details on this construction can be found in \cite{dashti17,stuart_16}.} Due to the assumption of independence, the joint prior can be written as:
\begin{align}\label{eq:4.8}
    \mu^{0}(du_1,du_2)=\mu_1^0(du_1)\times \mu_2^0(du_2).
\end{align}Note that the pushforward measures, $H_1^*(\mu_1^0)$ and $H_2^*(\mu_2^0)$ are prior measures on $\Ac(\Omega)$ for parameters $a(x)$ and $b(x)$ respectively. We will specify the particular covariance kernels and the means in the section on numerical simulations.
\section{Likelihood and posterior update}\label{sec5}
Concatenating the observations given by equation \eqref{eq:4.6} for each $y_j, j\in \{1,\dots,J\}$ as a single vector $y$, we get the analogue of equation \eqref{eq:2.7}:
\begin{align}\label{eq:5.1}
y&= \Pi \circ M(H_1(u_1),H_2(u_2))+\eta,\quad  \eta \sim N(0,\Gamma):=\Qb_0\nonumber\\
&=G(u_1,u_2)+ \eta
\end{align}
where $\Gamma=diag(\Gamma_0,\dots,\Gamma_0)$ as before. We rewrite $\Pi \circ M(H_1(u_1),H_2(u_2)):=G(u_1,u_2)$ for brevity. Now we will follow the program introduced for Bayesian level-set approach to inverse problems in \cite{stuart_16,stuart16b,stuart16a}. First of all, we will show that given the prior measure $\mu^0$ introduced in section \ref{prior}, there exists a posterior measure $\mu^y$ whose density with respect to the prior measure takes the following form:
\begin{align}
\frac{du^{y}}{d\mu^0}(u_1,u_2)=\frac{1 }{Z(y)} \exp(-\Phi(u_1,u_2;y))
\end{align}
where $\Phi(u_1,u_2;y)$ is a real valued function (negative log-likelihood) given by,
\begin{align}\label{eq:5.3}
\Phi(u_1,u_2;y)=\frac{1}{2}\lvert y-G(u_1,u_2)\rvert_{\Gamma}^2:=\frac{1}{2}\lvert \Gamma^{-1/2} (y-G(u_1,u_2))\rvert^2.
\end{align}
and $Z(y)$ is a normalizing factor given as:
\begin{align}\label{eq:5.4}
Z(y)=\int_{C(\bar{\Omega})\times C(\bar{\Omega})} \exp (\Phi(u_1,u_2;y))\mu^0(du_1,du_2)
\end{align}
In fact, to show that such an infinite-dimensional version of the Bayes' theorem holds in our setting is equivalent to establishing certain regularity properties for the negative log-likelihood function as encapsulated by \cite[Proposition 3.10]{stuart_16}. We prove an analogue of that proposition here, adapting it for the two-parameters ($a$ and $b$) case wherever needed.

\begin{lemma}\label{prop5.1}
  The negative log likelihood function $\Phi(u_1,u_2;y)$ given by \eqref{eq:5.3} satisfies the following properties:
  \begin{enumerate}
      \item There exists a positive, continuous function $K(\cdot)$ such that whenever $\lvert y \rvert_{\Gamma}<\rho$ for any given $\rho$, we have $0\leq \Phi(u_1,u_2;y)\leq K(\rho).$
      \item For any fixed $y$, the map $\Phi(\cdot,\cdot;y)$ is continuous $\mu^0$-almost surely.
      \item There exists a positive continuous function $\tilde{K}(\cdot)$ such that if $\max\{\lvert y_1\rvert_{\Gamma},\lvert y_2\rvert_{\Gamma}\}\leq \rho$ for any given $\rho$, then 
      \begin{align}
          \lvert \Phi(u_1,u_2;y_1)-\Phi(u_1,u_2;y_2)\rvert \leq \tilde{K}(\rho) \lvert y_1-y_2\rvert_{\Gamma}.
      \end{align}
  \end{enumerate}  
\end{lemma}
 \begin{proof}
 Using similar notation as introduced in the proof of Theorem \ref{thm:3.6}, we will write 
$v:=\Mca_j(a_1,b_1)$ and $v_{\epsilon}:=\Mca_j(a_2,b_2)$. Thus $ \lVert \Mca_{j}(a_1,b_1)-\Mca_j(a_2,b_2)\rVert_{*}=\lVert v-v_{\epsilon}\rVert_{*}$. Also, recall that $a_1,a_2,b_1$ and $b_2$ are all bounded functions. Thus, from \eqref{eq:3.24}, we have: 
\begin{align}\label{eq:5.6}
 \lVert \Mca_{j}(a_1,b_1)-\Mca_j(a_2,b_2)\rVert_{*}&\lesssim \bigg[\bigg(\int_{\Omega} \lvert (a_{1}-a_2)\rvert^2 \lvert \nabla v\rvert^2 dx\bigg)^{\frac{1}{2}}+\bigg(\int_{\Omega}\lvert(b_{1}-b_2)\rvert^2 \lvert v\rvert^2 dx\bigg)^{\frac{1}{2}} \bigg]\nonumber\\
 &\lesssim \bigg[\bigg(\lVert a_1-a_2\rVert_{\infty}\int_{\Omega}  \lvert \nabla v\rvert^2 dx\bigg)^{\frac{1}{2}}+\bigg(\lVert b_1-b_2\rVert_{\infty}\int_{\Omega} \lvert v\rvert^2 dx\bigg)^{\frac{1}{2}} \bigg]\nonumber\\
  &\lesssim \bigg[\lVert a_1-a_2\rVert_{\infty}\lVert v\rVert_{*}+\lVert b_1-b_2\rVert_{\infty}\lVert v\rVert_{*}\bigg]\quad {\text{(from definition of }\lVert \cdot \rVert_{*})}\nonumber\\
   &= \bigg[\lVert a_1-a_2\rVert_{\infty}\lVert +\lVert b_1-b_2\rVert_{\infty}\lVert \bigg]\lVert \Mca_j(a_1,b_1)\rVert_{*}
\end{align}
If we fix values for $a_1=a_c$ and $b_1=b_c$, the inequality \eqref{eq:5.6} can be rewritten as:
\begin{align}
\lVert \Mca_j (a_2,b_2)\rVert_{*}\lesssim \lVert \Mca_j(a_c,b_c)\rVert_{*}+\bigg[\lVert a_2\rVert_{\infty}\lVert +\lVert b_2\rVert_{\infty}\lVert \bigg]\lVert \Mca_j(a_c,b_c)\rVert_{*}
\end{align}
Hence, for any $(a,b)$ such that $\lVert a\rVert_{\infty}+\lVert b\rVert_{\infty}< M$, we get:
\begin{align*}
 \lVert \Mca_j(a,b)\rVert_{*}<C(1+M)
\end{align*}
where $C=\max_{j}\lVert\Mca_j(a_c,b_c)\rVert_{*}.$ From this and Theorem \ref{thm:3.6}, it also follows that the map $\Pi\circ \Mca_j (a,b)$ is bounded for all $j$ whenever $(\lVert a\rVert_\infty +\lVert b\rVert_{\infty})$ is bounded. From the definition of negative log-likelihood function $\Phi(u_1,u_2;y)$ we have: 
\begin{align*}
    \Phi(u_1,u_2;y)\leq \lvert G(u_1,u_2)\rvert_{\Gamma}^2+\lvert y\rvert_{\Gamma}^2.
\end{align*}
Now in our setting we have that $H_1(u_1)=a(x)$ and $H_2(u_2)=b(x)$ are uniformly bounded. Thus $\Pi_j\circ\Mc_j(H_1(u_1),H_2(u_2))$ is bounded above for every $j$ and let $\lvert G(u_1,u_2)\rvert_{\Gamma}\leq \tilde{C}=\max_{j}\lvert \Pi_j\circ\Mc_j(H_1(u_1),H_2(u_2))\rvert_{\Gamma}$. Thus if $\lvert y\rvert_{\Gamma}^2<\rho$, then $\Phi(u_1,u_2;y)<\tilde{C}^2+\rho^2:K(\rho)$. This proves (1).
\par  Let $(u_1,u_2)\sim \mu^0$, i.e. $u_1\sim \mu^0_1$ and $u_2\sim \mu^0_2$. Consider two sequences, $(u_{1_\epsilon})$ and $(u_{2_\epsilon})$ such that $\lVert u_{1_\epsilon}-u_1\rVert_{\infty}\to 0$ and $\lVert u_{2_\epsilon}-u_2\rVert_{\infty}\to 0$. From proposition \ref{prop4.1}, it follows that $H_1(u_{1_{\epsilon}})$ and $H_2(u_{2_{\epsilon}})$ converge respectively to $H_1(u_1)$ and $H_2(u_2)$ in measure as long as the level-sets for the functions $u_1,u_2$ have zero measure. As these functions are drawn from a Gaussian measure, this property is $\mu^0$ almost surely true for $(u_1,u_2)$. Now from Theorem \ref{thm:3.6}, we get $G_j=\Pi\circ\Mc_j(H_1(\cdot),H_2(\cdot))$ is $\mu^0$-almost surely continuous at any $(u_1,u_2)$. This proves (2).
\par Fix $u_1$ and $u_2$ and consider $y_1$ and $y_2$ such that $\max\{\lvert y_1\rvert_{\Gamma},\lvert y_2\rvert_{\Gamma}\}<\rho$,
 \begin{align*}
       \lvert \Phi(u_1,u_2;y_1)-\Phi(u_1,u_2;y_2)\rvert &=   \bigg\lvert\frac{1}{2}\lvert y_1-G(u_1,u_2)\rvert_{\Gamma}^2-\frac{1}{2}\lvert y_2-G(u_1,u_2)\rvert_{\Gamma}^2\bigg\rvert \nonumber\\
       &=\frac{1}{2}\bigg\lvert\lvert y_1\rvert_{\Gamma^2}-2\langle G(u_1,u_2), y_1\rangle_{\Gamma}-\lvert y_2\rvert_{\Gamma}^2+2\langle G(u_1,u_2),y_2\rangle_{\Gamma}\bigg\rvert \\
       &=\frac{1}{2}\bigg\lvert \langle y_1+y_2-2G(u_1,u_2),y_1-y_2\rangle_{\Gamma}\bigg\rvert\\
       &\leq (\rho+\lvert G(u_1,u_2)\rvert_{\Gamma})\lvert (y_1-y_2)\rvert_{\Gamma}\\
       &\leq (\rho+\tilde{C}) \lvert (y_1-y_2)\rvert_{\Gamma}
\end{align*}
where we use the fact that $\lvert G(u_1,u_2)\rvert_{\Gamma}\leq \tilde{C}$ as in the proof of [Proposition \ref{prop5.1},(1)] above. Let $\tilde{K}(\rho):=\tilde{C}+\rho$. This proves (3).
\end{proof}
\begin{theorem}[Main Theorem] \label{thm:5.2}
Consider the space $\Ac(\Omega)\times \Ac(\Omega)$ on which a prior measure $\mu^0$ introduced in \eqref{eq:4.8} is given. Consider the potential function $\Phi(u_1,u_2;y)$ and the normalizing factor $Z(y)$ given by \eqref{eq:5.3} and \eqref{eq:5.4} respectively. Then the posterior measure $\mu^y$ on the space $\Ac(\Omega)\times \Ac(\Omega)$ exists and is absolutely continuous with respect to the prior measure. In fact, its Radon-Nikodym derivative is given by:
\begin{align}\label{eq:5.8}
\frac{d\mu^y}{d\mu^0}(u_1,u_2)=\frac{1 }{Z(y)} \exp(-\Phi(u_1,u_2;y)).
\end{align}
 Furthermore, the posterior measure is Lipschitz with respect to the observed data in the following sense: for any $y,y^{\prime}$ with $\max\{\lvert y\rvert_{\Gamma}, \lvert y^{\prime}\rvert_{\Gamma}\}<\rho$, there exists a function depending upon $\rho$, C:=$C(\rho)$ such that,
 \begin{align}\label{eq:5.9}
 \lvert \mu^{y^{\prime}}-\mu^y\rvert_{\rm{Hell}}\leq C \lvert y^{\prime}-y\rvert_{\Gamma}
 \end{align}
 Here $\lvert p_1-p_2\rvert_{\rm{Hell}}$ denotes the Hellinger distance between two probability measures $p_1$ and $p_2$.
\end{theorem}
\begin{proof}
We will follow the program laid out in the proof of \cite[Theorem 2.3]{stuart16a}. The main idea is to establish the continuity of the function $\Phi(u_1,u_2;y)$ jointly under the three variables and prove a lower bound for the normalizing factor $Z(y)$. This would allow us to claim the existence and well-posedness of the posterior measure $\mu^y$ as postulated in the theorem above, see \cite[Theorem 14]{dashti17}. 
\par Note that the random variable $y|(u_1,u_2)$ is distributed according to a measure $\mathbb{Q}_{u}:=N(G(u_1,u_2),\Gamma)$ which is obtained from a translation of the measure $\mathbb{Q}_{0}:=N(0,\Gamma)$ by $G(u_1,u_2)$, see \eqref{eq:5.1}. More over $\mathbb{Q}_{u}$ admits a Radon Nikodym derivative with respect to the measure $\mathbb{Q}_{0}$:
\begin{align*}
    \frac{d\Qb_{u}}{d\Qb_0}(y) \propto \exp(-\Phi(u_1,u_2;y)).
\end{align*}
Consider the following product measure:
\begin{align}
\nu^0(du_1,du_2,y):=\mu^0(du_1,du_2)\Qb_0(dy)=\mu^0_1(du_1)\mu^0_{2}(du_2)\Qb_0(dy).
\end{align}
We will first show that the function $\Phi(u_1,u_2;y)$ is almost surely $\nu^0-$continuous. To see this, consider a sequence of continuous functions, $(u_{1_n},u_{2_n},y_n)$ such that $\lVert u_{1_n}-u_1\rVert_{\infty}\to 0$, $\lVert u_{2_n}-u_2\rVert_{\infty}\to 0$ and $\lvert y_n-y\rvert_{\Gamma}\to 0$. Then we have:
\begin{align}\label{eq:5.11}
    \lvert \Phi(u_{1_n},u_{2_n},y_n) - \Phi(u_1,u_2,y)\rvert\leq  \lvert \Phi(u_{1_n},u_{2_n},y_n)& - \Phi(u_{1_n},u_{2_n},y)\rvert\\&+ \lvert \Phi(u_{1_n},u_{2_n},y) - \Phi(u_1,u_2,y)\rvert\nonumber.
\end{align}
Now the second term in \eqref{eq:5.11} goes to zero $\mu^0-$ almost surely by [Proposition \ref{prop5.1}- (2)]. On the other hand, the first term goes to 0 by  [Proposition \ref{prop5.1}- (3)]. Thus $\Phi(u_1,u_2;y)$ is $\nu^0-$ almost surely jointly continuous in $u_1,u_2$ and $y$.
Next, we want to establish the bounds for $Z(y)$. Consider the product space, $C(\Omega)\times C(\Omega)=:X$. Fix any $\rho>\lvert y\rvert_{\Gamma}$ and consider the $K(\rho)$ as defined in [Proposition \ref{prop5.1}- (1)].
Thus we have:
\begin{align}
0\leq &\Phi(u_1,u_2;y)\leq K(\rho)\nonumber\\
\implies\int_{X}d\mu^{0}(u_1,u_2)\geq \int_{X} \exp(-\Phi(u_1,&u_2;y))d\mu^{0}(u_1,u_2)\geq \int_{X} \exp(-K(\rho))d\mu^{0}(u_1,u_2)\nonumber\\
\implies 1&\geq Z(y)\geq \exp(-K(\rho))>0\label{line536}
\end{align}
Thus it follows from \cite[Theorem 14]{dashti17}, that $\mu^y$ exists and satisfies \eqref{eq:5.8}. To prove the well-posedness of the bayesian inverse problem, we need to establish Lipschitz continuity of the posterior measure $\mu^y$ with respect to observed data $y$. To that end, first recall that the Hellinger distance between two measures $p_1$ and $p_2$ defined on some space $Y$ and both absolutely continuous with respect to some measure $p_0$ is given by:
\begin{align}
\lvert p_1-p_2\rvert_{\rm{Hell}}^2=\frac{1}{2}\int_{Y}\bigg(\sqrt{\frac{dp_1}{dp_0}}-\sqrt{\frac{dp_1}{dp_0}} \bigg)^2 dp_0.
\end{align}
Thus if we have two posterior measures $\mu^y$ and $\mu^{y^{\prime}}$, both absolutely continuous with respect to the prior measure $\mu^0$, then
\begin{align}\label{eq:5.13}
2\lvert \mu^y-\mu^{y^{\prime}}\rvert^2&=\int_{X} \bigg(\sqrt{\frac{\exp(-\Phi(u_1,u_2;y))}{Z(y)}}-\sqrt{\frac{\exp(-\Phi(u_1,u_2;y^{\prime}))}{Z(y^{\prime})}} \bigg)^2 d\mu_0 \nonumber\\
&\leq I_1+ I_2
\end{align}
where \begin{align*}
I_1&=\frac{2}{Z(y)}\int_{X} \bigg(\sqrt{\exp(-\Phi(u_1,u_2;y))}-\sqrt{\exp(-\Phi(u_1,u_2;y^{\prime}))} \bigg)^2 d\mu_0 \quad \quad \text{and}\\
I_2&=\frac{2}{\lvert \sqrt{Z}-\sqrt{Z^{\prime}}\rvert^2}\int_{X}\exp(-\Phi(u_1,u)2;y^{\prime}))d\mu^0 \quad (Z(y^{\prime}):=Z^{\prime})
\end{align*}
Note that in writing \eqref{eq:5.13} we have used the fact that, $(ab-cd)^2\leq 2a^2(b-d)^2+2(a-c)^2d^2$.
Similar to the proof of \cite[Theorem 2.3] {stuart16a}, one can conclude from Proposition \ref{prop5.1} and the fact that $Z(y)$ is bounded, that $I_1\leq C_1 \lvert y-y^{\prime}\rvert_{\Gamma}^2$ and $I_2\leq C_2\lvert y-y^{\prime}\rvert_{\Gamma}^2$ for some $C_1$ and $C_2$. In order to see this, notice 
\begin{align}
\frac{Z(y)}{2}I_1 &= \int_{X}\bigg({\exp(-\frac{1}{2}\Phi(u_1,u_2;y))}-{\frac{1}{2}\exp(-\Phi(u_1,u_2;y^{\prime}))} \bigg)^2 d\mu_0 \nonumber\\
&\lesssim \int_{X} \bigg({(-\Phi(u_1,u_2;y))}-{(-\Phi(u_1,u_2;y^{\prime}))} \bigg)^2 d\mu_0 \quad \text{(from \ref{prop5.1} (1 and 2))}\nonumber\\
&\lesssim \lvert y -y^{\prime}\rvert_{\Gamma} ^2 \quad \text{(from \ref{prop5.1} (3)) }.\label{eq:5.14a}
\end{align}
In addition, 
\begin{align}
I_2&\lesssim {\lvert {Z}^{-\frac{1}{2}}-{{Z^{\prime}}^{-\frac{1}{2}}}\rvert^2} \quad \text{(from \ref{prop5.1} (1))}\nonumber \\
&={\bigg\lvert \frac{ Z-Z^{\prime}}{(\sqrt{Z}+\sqrt{Z^{\prime}})\cdot \sqrt{Z\cdot Z^{\prime}}} \bigg\rvert^2} \lesssim \lvert Z-Z^{\prime} \rvert^2 \label{eq:5.14}
\end{align}
Note that, in \eqref{eq:5.14}, for the second inequality, we use that fact that $Z$ and $Z^{\prime}$ are bounded above by $1$ and bounded below by a strictly positive number, $\exp(-K(\rho))>0$, see \eqref{line536}. 
Finally, note that 
\begin{align}
   \lvert Z-Z^{\prime}\rvert =&\bigg\lvert\int_{X} \exp(-\Phi(u_1,u_2;y))d\mu^{0}- \exp(-\Phi(u_1,u_2;y^{\prime}))d\mu^{0}\bigg \rvert \\
   &\lesssim \lvert y-y^{\prime}\rvert_{\Gamma} \quad \text{(from proposition \ref{prop5.1}, similar to eqn. \eqref{eq:5.14a}}).
\end{align} 

This establishes \eqref{eq:5.9} and hence concludes the proof of the theorem. 
\end{proof}
\section{Numerical Simulations}
In this section we provide numerical validation of the theoretical results using simulated experiments. The data has been generated by solving the forward problem on a finer finite-element mesh with 2129 elements while the inversion has been carried out on a coarser mesh with 549 elements to avoid the inverse crime. For illustration, we give both these meshes in Figure \ref{mesh}. In the next subsection we describe the pCN sampling approach that has been used to generate samples from the posterior distribution of diffusion and absorption parameters. Subsequently, we provide reconstructions using two different methods and compare the results. {We note that all the phantoms considered below are bi-level phantoms, with background values for diffusion, $a_{\text{back}}=1$ and for absorption, $b_{\text{back}}=0.1$ (in first and third phantoms) or $b_{\text{back}}=0.2$ (as in the second phantom below). The anomaly levels for the two parameters will be $a_{\text{fore}}=5*a_{\text{back}}$ and $b_{\text{fore}}=5*b_{\text{back}}$, respectively.} { We have also experimented with other values for the background and inhomogeneity and the simulation results are similar. {All simulations were done on an x-64 based PC with 1.60 GHz, 4-core i5 processor (8th gen)}.}
\begin{figure}[ht]
	\centering
\includegraphics[width=0.4\textwidth]{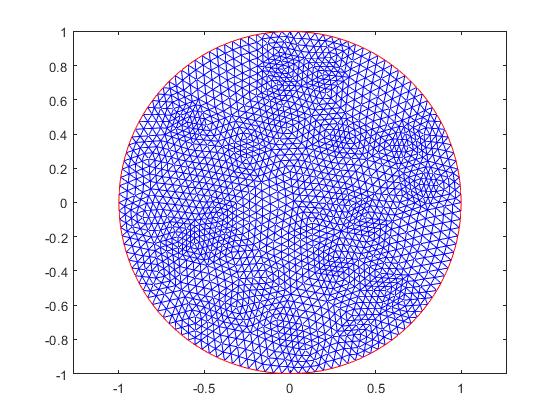}\hspace{0.01cm}
 \includegraphics[width=0.4\textwidth]{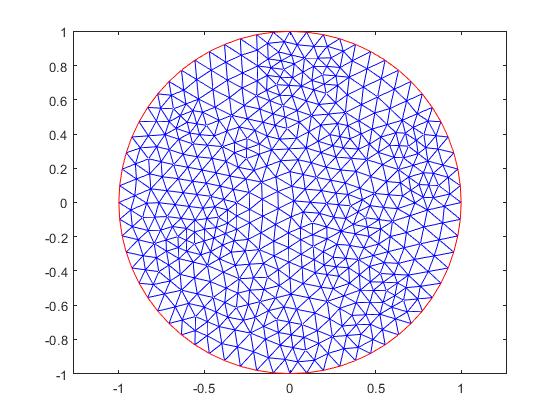}\hspace{0.01cm}
 \caption{On the left is the fine mesh used to generate the data by solving the forward problem. On the right is the coarser finite element mesh used for solving the inverse problem.}
 \label{mesh}
 \end{figure}
\subsection{Markov Chain Monte Carlo} 
In this subsection we illustrate how to obtain samples from the posterior distribution for each of the parameters using an MCMC (Markov Chain Monte Carlo) method. We recall some notation from section \ref{prior}. We denote by $\mu^0$, the joint prior measure on the space of parameters such that $\mu^0=\mu^0_1\times \mu^2_0$ where $\mu^0_i,\ i=\{1,2\}$ are Gaussian priors with zero mean and covariance $C_i,\ i=\{1,2\}$ corresponding to the parameters $a$ and $b$ respectively. The choice for each $C_i$ will be made precise later. For brevity's sake we denote the tuple $(u_1,u_2):=\rm{f}$, where $H_1(u_1)=a,$ and $H_2(u_2)=b$. We fix the following notation: for any random function (variable) $X$, we denote by $X^*$ the corresponding discretization of $X$ relative to the finite element mesh used for the inverse problem. We will now describe the pCN sampling method (see, {\cite{cotter_13}}) used to generate posterior samples in this work. First, define the log-likelihood function for the DOT problem for the $k-$th sample, $f^*_k:=(u^*_{1_k},u^*_{2_k})$ where $k\in \{1, 2, ..., N\}$,
\begin{equation}
L(f^*_k):= -\frac{1}{2\sigma_{noise}^2} ||y - {G}^{*}(f_k)||_{F}^2,
\end{equation}
where $y$ is the measured data, ${G}^{*}(f_k)$ refers to the solution of the forward problem with parameter $f^*_k$, and $||\cdot||_F^2$ is the squared Frobenius norm. 

\noindent Set $f_0^*=(a^*_0,b^*_0)$ where $a_0^*$ and $b_0^*$ are equal to the typical background value of $a^*$ and $b^*$ respectively. Then we iterate the following steps until the desired number of samples is obtained:
\begin{enumerate}
\item Draw $\Upsilon\sim\mu^0$ and determine the proposals:
$f^*_{{\text{Prop.}}} :=\sqrt{1-2\Delta} f^*_k + \sqrt{2\Delta} \Upsilon$, for {$\Delta>0$}\\
\item Set $f^*_{k+1} = \begin{cases}
f^*_{\text{Prop.}} & \text{with probability } \min(1, L(f^*_{\text{Prop.}}) - L(f^*_k)), \\
f_k^* & \text{else}.
\end{cases}  $ 
\end{enumerate}

Notice that one requires a significant amount of burn-in iterations before reaching the high probability areas of the posterior distribution, such that the samples drawn will be representative of the posterior distribution. For fast convergence, it is recommended to select the acceptance rate $q>0$ so that after burn in approximately $25\%$ of the proposed samples are accepted. In practice, this problem is not trivial. Hence, we adapt $\Delta$ during the burn-in time to stabilize the acceptance rate close to $0.25$. The pCN algorithm  used in this work, has robust convergence properties with respect to the (very high) number of dimensions used in the finite dimensional approximation of an infinite dimensional parameter space of functions to which $a(x)$ and $b(x)$ belong. We refer interested readers to \cite{cotter_13} for details on how the pCN approach is well-suited for implementation of non-parametric Bayesian estimation methods in function-space settings.

Now we describe the choice for the covariance (kernel) functions $C_i, \ \{i=1,2\}.$ To obtain proper regularized reconstructions, we choose the covariance matrices to be defined by a Matérn kernel, with parameters $\nu_i$ and $\ell_i$ chosen heurestically. The Matérn kernels $C_i$ are given by 
\begin{equation}
C_i=C_{\nu_i,\ell_i}(d):=\frac{2^{1-\nu_i}}{\Gamma(\nu_i)}\left( \frac{d\sqrt{2\nu_i}}{\ell_i}\right)^{\nu_i} K_{\nu_i}\left( \frac{d\sqrt{2\nu_i}}{\ell_i}\right),
\end{equation}
where $d:=||x_k-x_j||_2$ is the euclidean distance between two PDE mesh centroids and $K_{\nu_i}$ is the modified Bessel function of the second kind. For reconstructions shown below, we choose $l_i=0.3, \{i=1,2\}$ while $\nu_1=4$ and $\nu_2=5$. The question of properly choosing these hyperparameters $\nu_i$ and $l_i$ is a delicate one; and seems to depend on the size of inclusions. {{Taking inspiration from the finite-dimensional settings, one may attempt a heirarchial Bayesian model to determine the ideal choice of these hyperparameters. Extending this idea to a non-parametric infinite-dimensional setting has garnered an increased interest, see \cite{sui_24,dunlop_17,dunlop2019} . However, we leave the choice of hyperparameter determination in our context for a future work.}}  The value of $\Delta$ was chosen as $0.0025$ to begin with and adapted during burn-in time so as to stabilize the acceptance ratio to be around $25\%$. The MCMC chain was allowed to run for $300,000$ iterations with $50,000$ samples discarded as burn-in. This forms one experiment with each phantom. Reconstruction results based on one such experiment for each phantom are given in the next subsection. {In our diagnostics for the MCMC chain, we have looked at the traceplots and it shows good mixing properties that can be visualized in Figure \ref{mix}.} Ten experiments each at relative noise levels of $1\%-4\%$ were done to test the robustness of reconstructions to noise. Results from this study are given in the subsequent subsection.
\begin{figure}[ht]
\centering
{\includegraphics[width=0.3\textwidth]{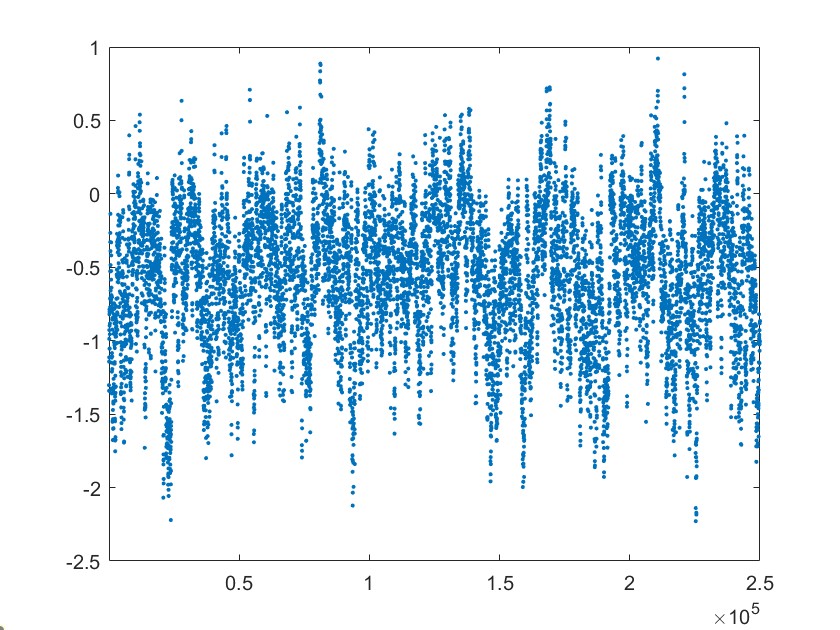}}
{\includegraphics[width=0.3\textwidth]{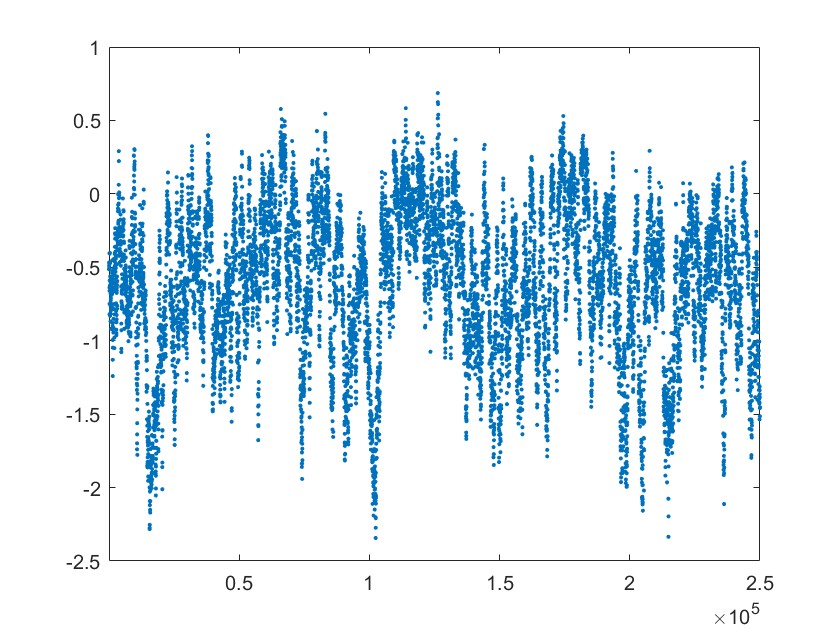}}
    {\includegraphics[width=0.3\textwidth]{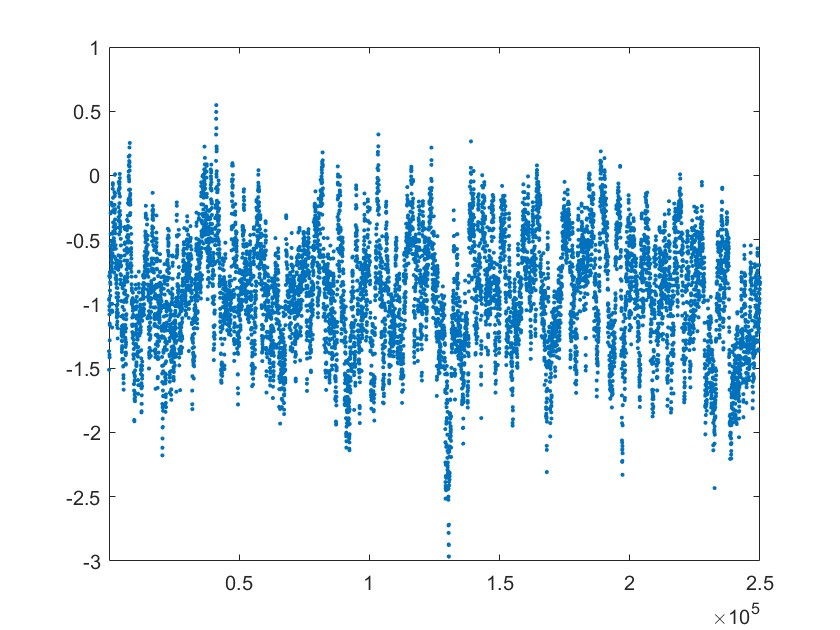}}
 {\includegraphics[width=0.3\textwidth]{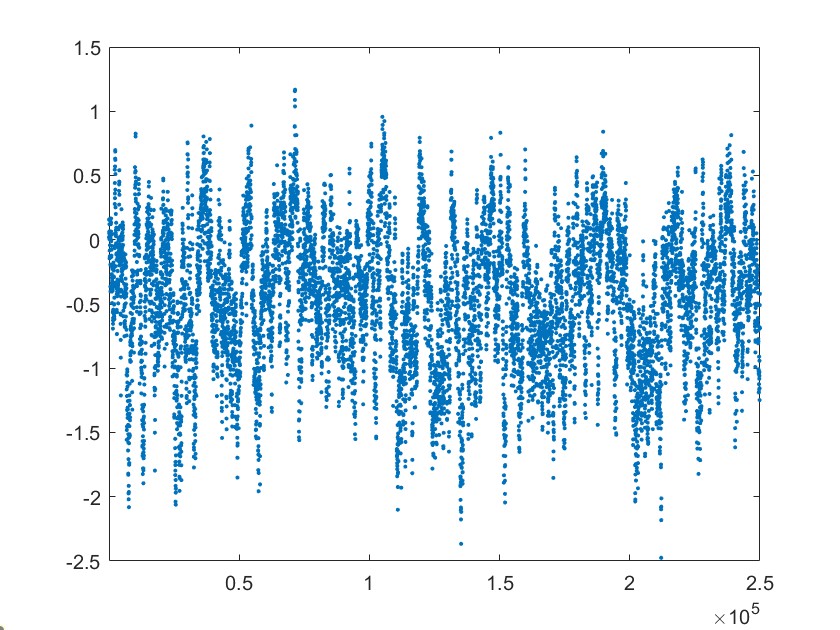}}{\includegraphics[width=0.3\textwidth]{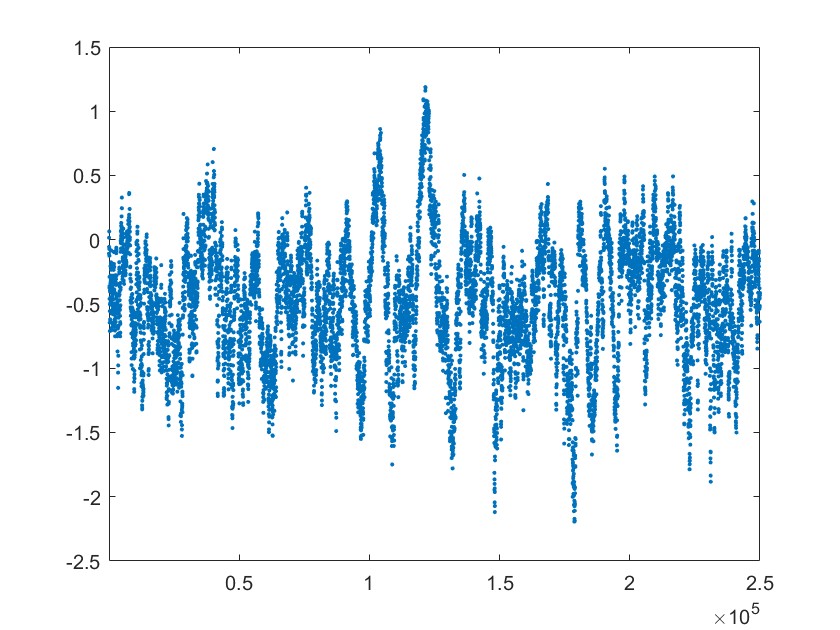}}
    {\includegraphics[width=0.3\textwidth]{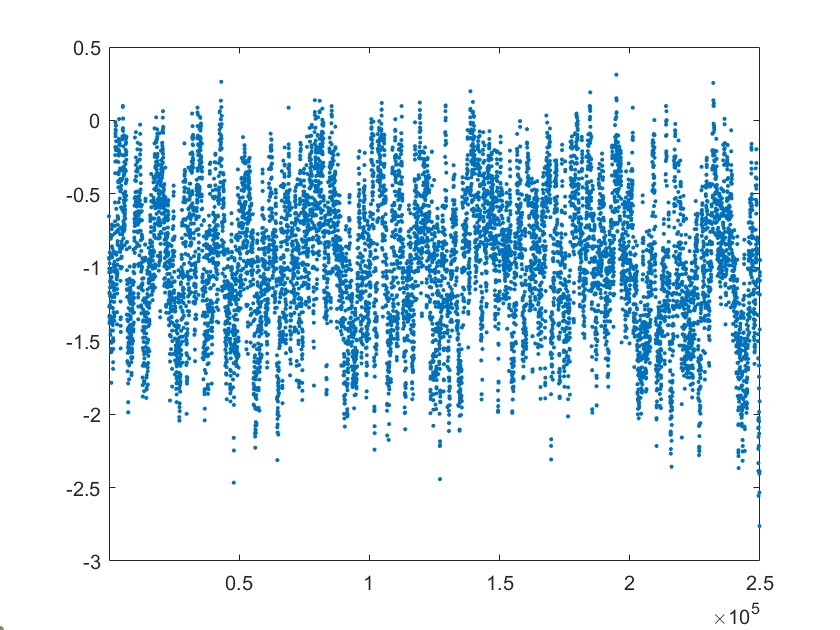}}
\caption{Traceplots for parameter values as a function of MCMC draws after burn-in: Row 1 corresponds to the traceplots for the value of $u_1$ (diffusive level-set function) in some randomly chosen triangle in the reconstruction mesh, Row 2 corresponds to the traceplots for the value of $u_2$ (absorptive level-set function) in some randomly chosen triangle in the reconstruction mesh. Columns 1, 2 and 3 correspond respectively to the the three different phantom geometries as in Figure \ref{mod1}.}
\label{mix}
\end{figure}

 \begin{remark}
 {In our simulations, we noticed some instabilities when only using the likelihood term. This is because the function $H(\cdot)$ does not see a numerical difference between small or large positive and negative values of $u_1$ and $u_2$. This instability can be easily circumvented by adding a penalty term to the likelihood $L_{reg}(f^*_k)=L(f^*_k) + \alpha_1 ||u_1||^2 + \alpha_2 ||u_2||^2$, where $\alpha_1$ and $\alpha_2$ are regularization parameters. {{We believe that the addition of this regularization term in the likelihood has the overall effect of using a Gaussian hybrid prior  to handle unknowns that cannot be well modeled by standard Gaussian
measures; this has been proposed in the literature before, see \cite{yao_16}.}}}
\end{remark}
\subsection{Simulation results}
We present here some reconstruction results obtained for bi-level phantoms with several different geometries. In each of the phantoms, the background regions are shown in blue and the anomalous regions are shown in red. Let us define the functions $H_i$ for $i=\{1,2\}$ in the following way:
\begin{align*}
{(H_1(u_1))(x)=\begin{cases} a_{\text{fore}}; \quad &\text{if }u_1(x)>0\\
a_{\text{back}}; \quad &\text{else}
\end{cases}}\text {  and} \quad {(H_2(u_2))(x)=\begin{cases} b_{\text{fore}}; \quad &\text{if }u_2(x)>0\\
b_{\text{back}}; \quad &\text{else}
\end{cases}}
\end{align*}
\begin{remark}{In actual practice, we use mollified versions of $H_1$ and $H_2$, see \cite [eqn (5.11)]{agh_11}.}
\end{remark}
\par After posterior samples of $u_1$ and $u_2$ are obtained, one can choose to reconstruct the actual parameters $a$ and $b$  using any of the following two methods.

\subsection{Method 1: Bi-Level Reconstruction}

Reconstruct using $H_i(\overline{u_i})$ for $i=1,2$, where $i=1$ reconstructs diffusion ($a(x)$) and $i=2$ reconstructs absorption parameter ($b(x)$). We also note that $\overline{u_i}$ is the mean of the posterior samples. This gives us sharp bi-level reconstructions as shown in Figure \ref{mod1}. In addition to the reconstructions of the ground-truth anomalous diffusive and absorptive regions, we also provide the corresponding credible regions in the coarse (reconstruction) mesh. {Recall, that in a 1-D setting, a Bayesian credible interval of size $1-\alpha$ is an interval $(c, d)$ such that $P(c \leq \tau \leq d | \{X_i\}_{i=1}^n) = 1-\alpha$, where $\tau$ is random variable and $\{X_i\}_{i=1}^n$ given samples of $\tau$. In this article, by a region of $1-\alpha$ credibility we refer to the 2-D extension of a Bayesian credible interval. A region of $1-\alpha$ credibility is defined on a set of non-overlapping triangles in the reconstruction mesh $\Omega_{discrete}$ covering $\Omega$, if for each triangle $t \in \Omega_{discrete}$ mesh we have an interval $(c(t), d(t))$ with $P(c(t) \leq \tau(t) \leq d(t) | \{X_i(t)\}_{i=1}^n) = 1-\alpha$. Here $\tau(t)$ is a random variable in the corresponding triangle. } The corresponding credible regions are shown below in Figure \ref{UQ}.

\begin{figure}
     \begin{center}
     \begin{tabular}{ c c  c c  }
      Truth `$a$' & Recon. `$a$' & Truth `$b$' & Recon `$b$' \\ 
{\includegraphics[width=0.22\textwidth]{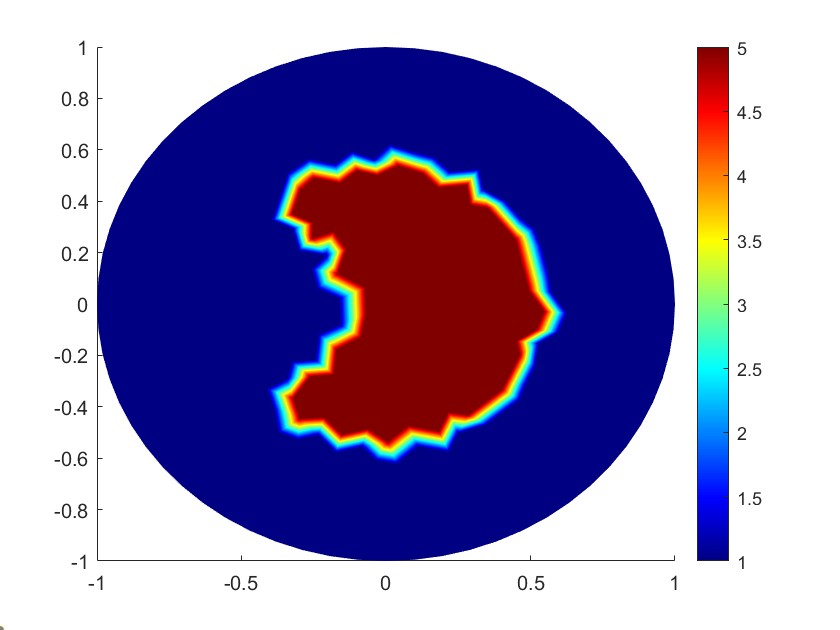}}&
    {\includegraphics[width=0.22\textwidth]{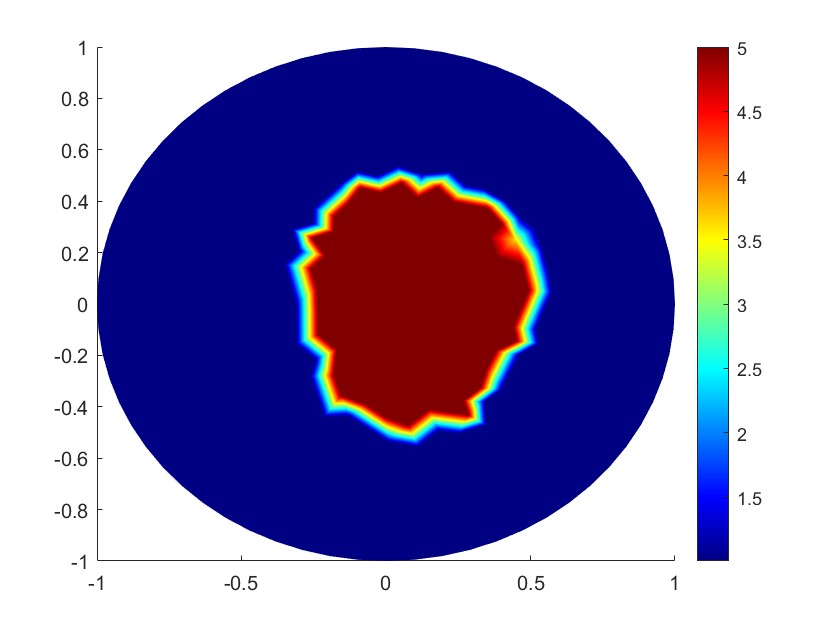}}&
  {\includegraphics[width=0.22\textwidth]{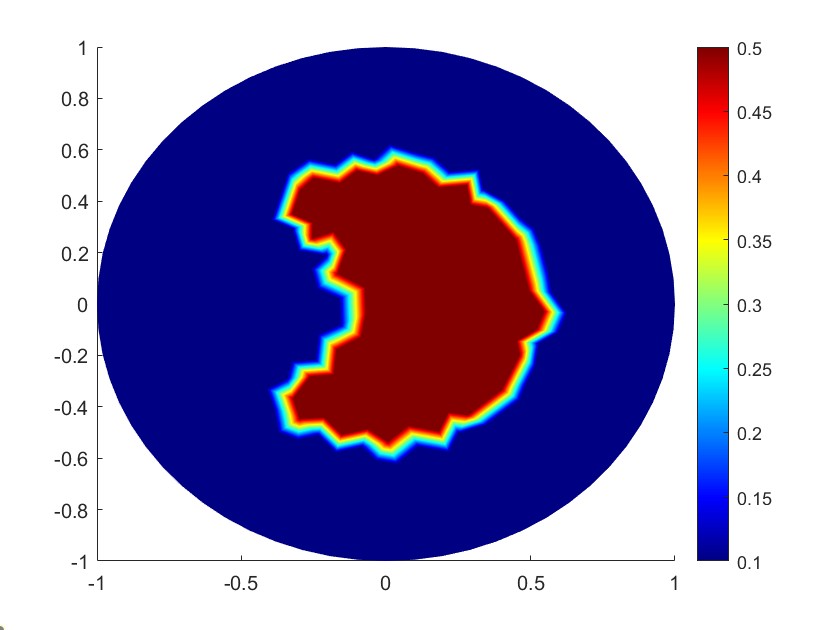}}&
    {\includegraphics[width=0.22\textwidth]{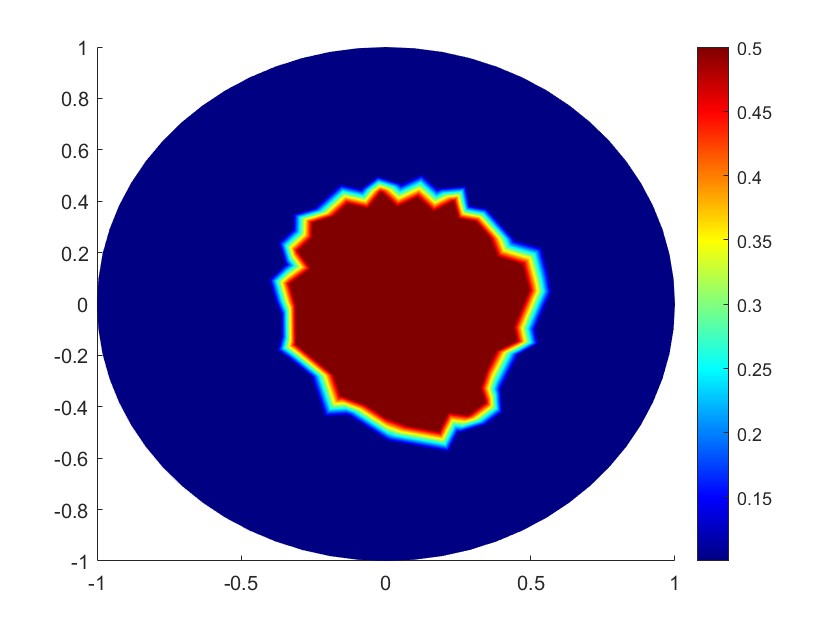}}\\
    {\includegraphics[width=0.22\textwidth]{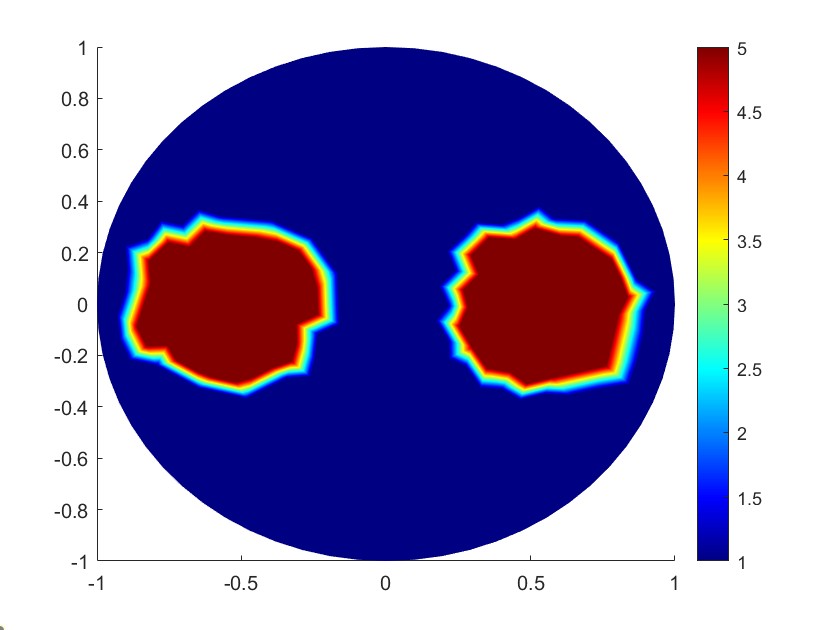}}&
    {\includegraphics[width=0.22\textwidth]{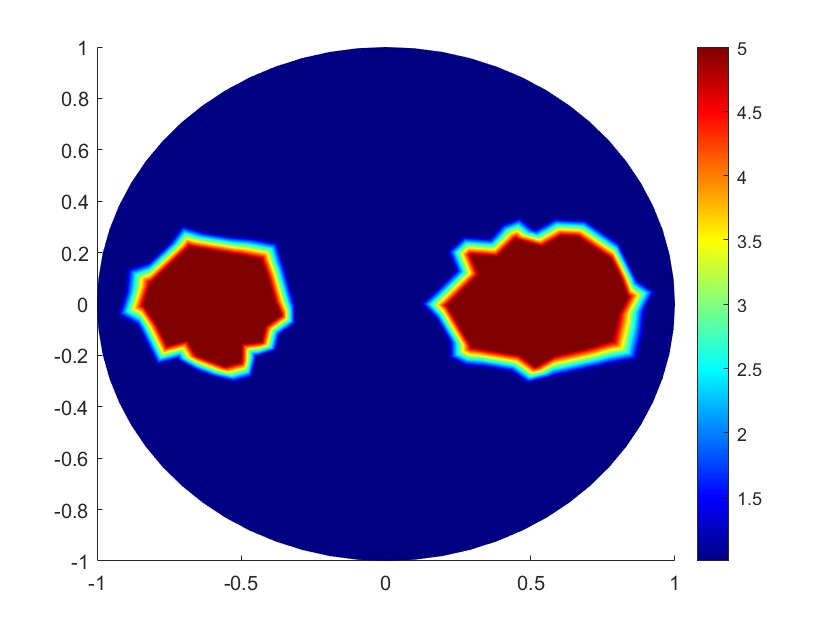}}&
  {\includegraphics[width=0.22\textwidth]{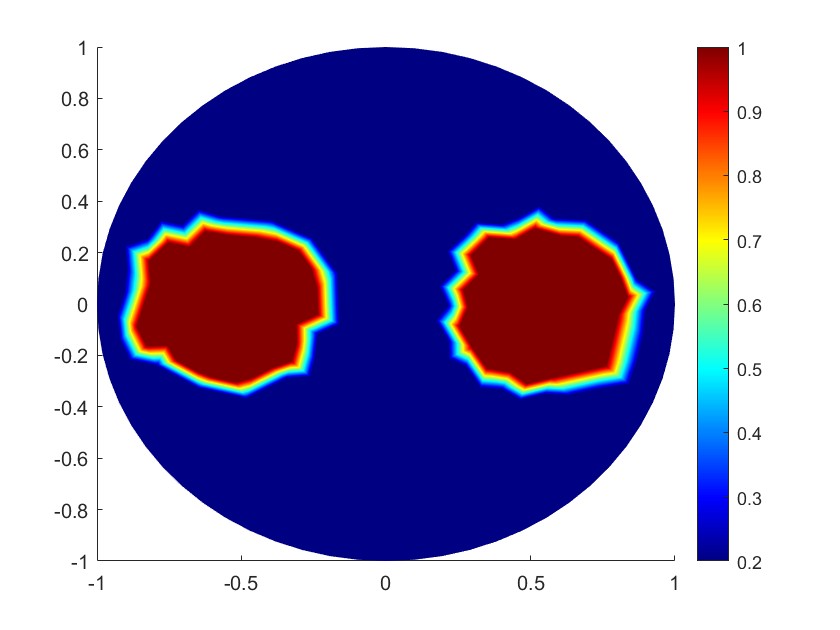}}&
    {\includegraphics[width=0.22\textwidth]{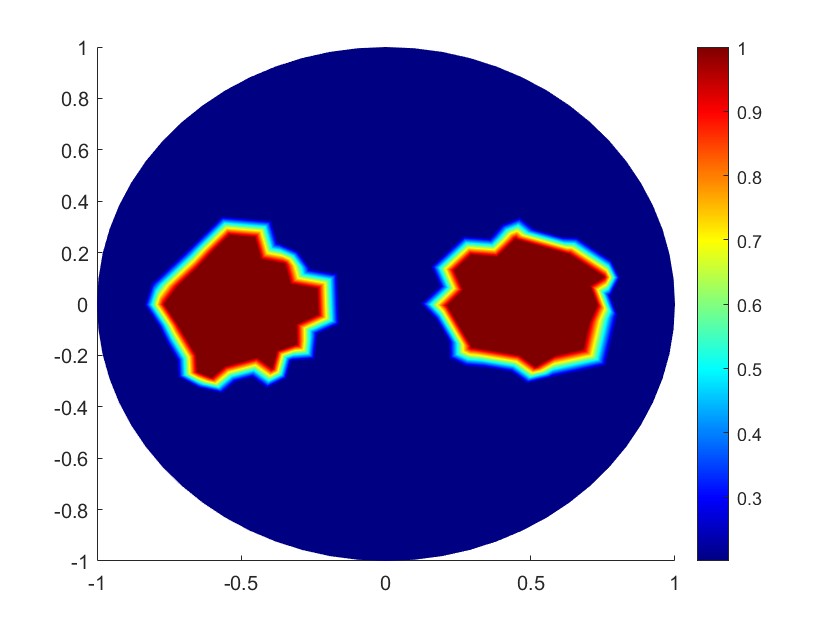}}\\
    {\includegraphics[width=0.22\textwidth]{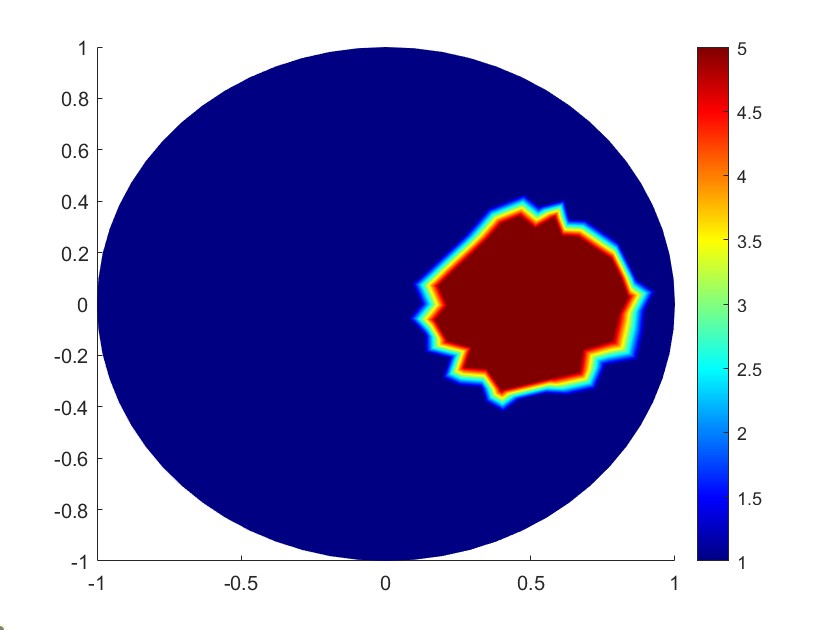}}&
    {\includegraphics[width=0.22\textwidth]{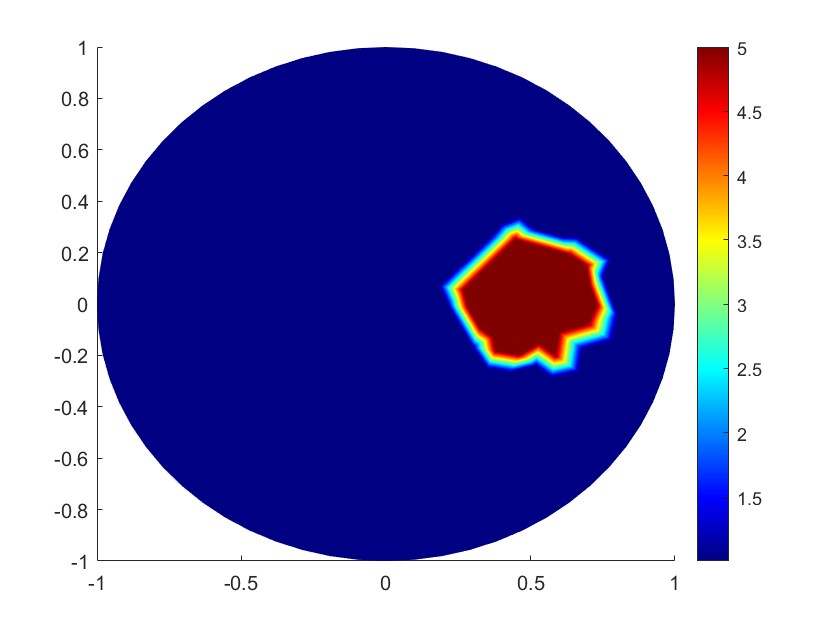}}&
  {\includegraphics[width=0.22\textwidth]{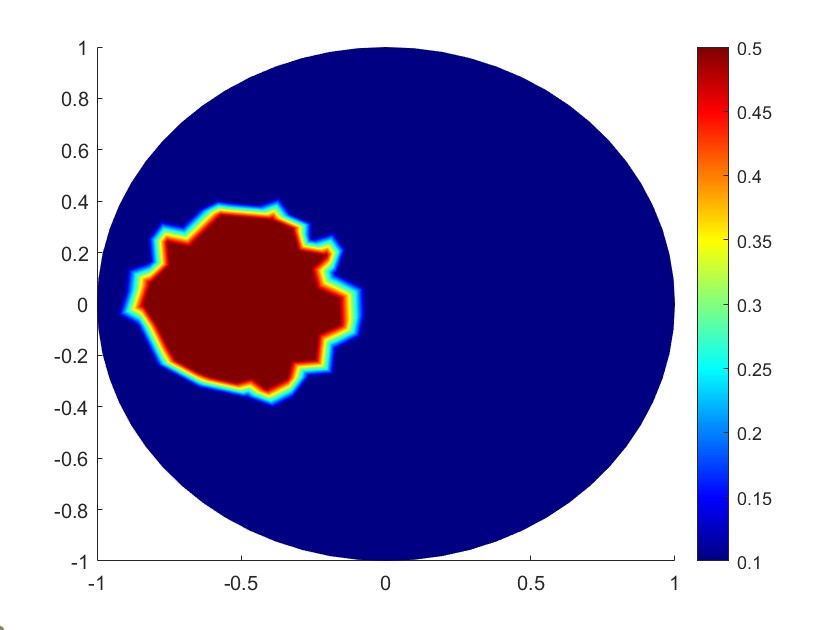}}&
    {\includegraphics[width=0.22\textwidth]{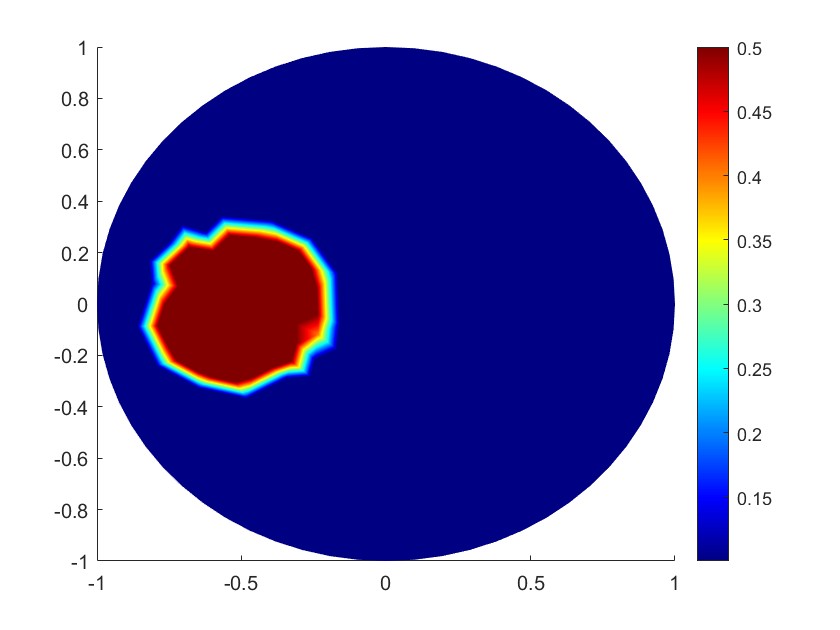}}

      \\ 
      \end{tabular}

\end{center}
            
\caption{Reconstructions using Method 1 (Bi-Level Reconstruction): First and third columns have the plots of { projections of true diffusion and absorption parameters on the FEM basis}. In rows 1 and 2, the diffusive and absorptive regions coincide. In row 3, we have a phantom in which the diffusive and absorptive regions are at separate places. In column 2, we have the corresponding reconstructions of the diffusive parameter and in column 4, we have the reconstructions of the absorptive parameter. The data has been collected at $2\%$ relative noise. }
 \label{mod1}
      \end{figure}


\begin{figure}
     \begin{center}
     \begin{tabular}{ c c  c c  }
      CR 15\% $a$ & CR 85\%  $a$ & CR 15\%  $b$ & CR 85\%  $b$ \\ 
    {\includegraphics[width=0.22\textwidth]{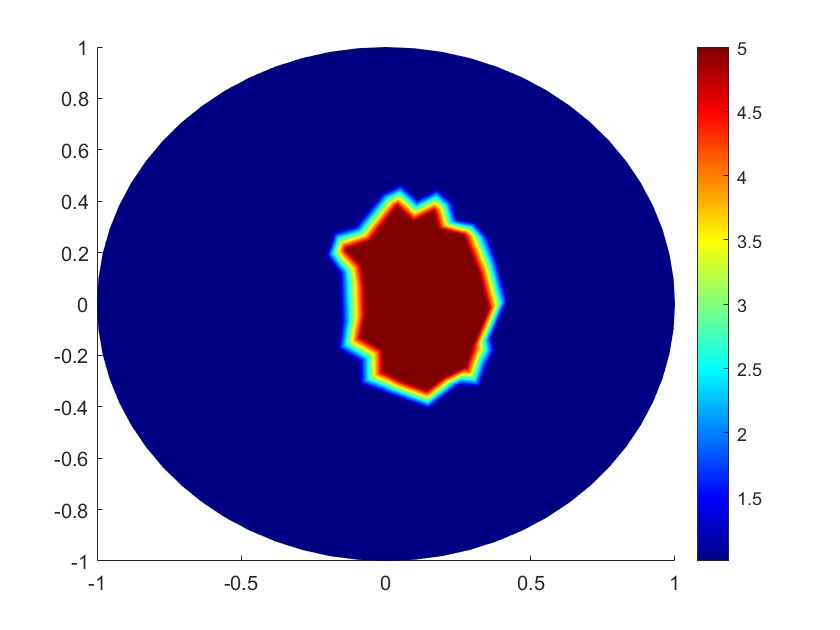}}&
    {\includegraphics[width=0.22\textwidth]{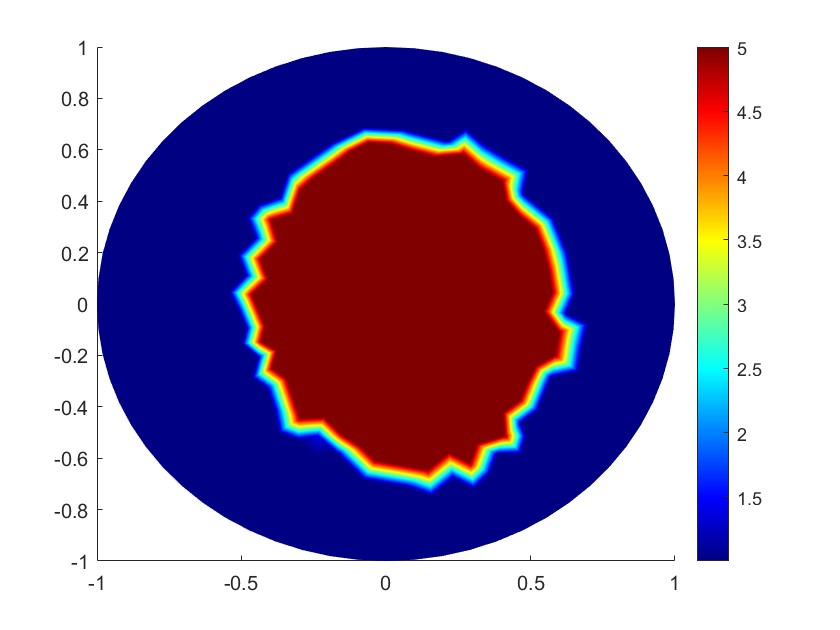}}&
  {\includegraphics[width=0.22\textwidth]{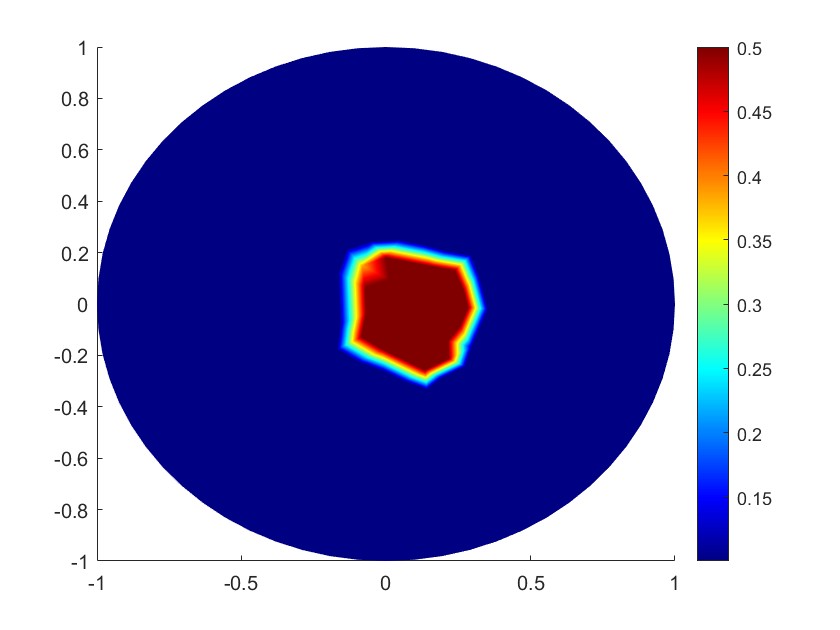}}&
    {\includegraphics[width=0.22\textwidth]{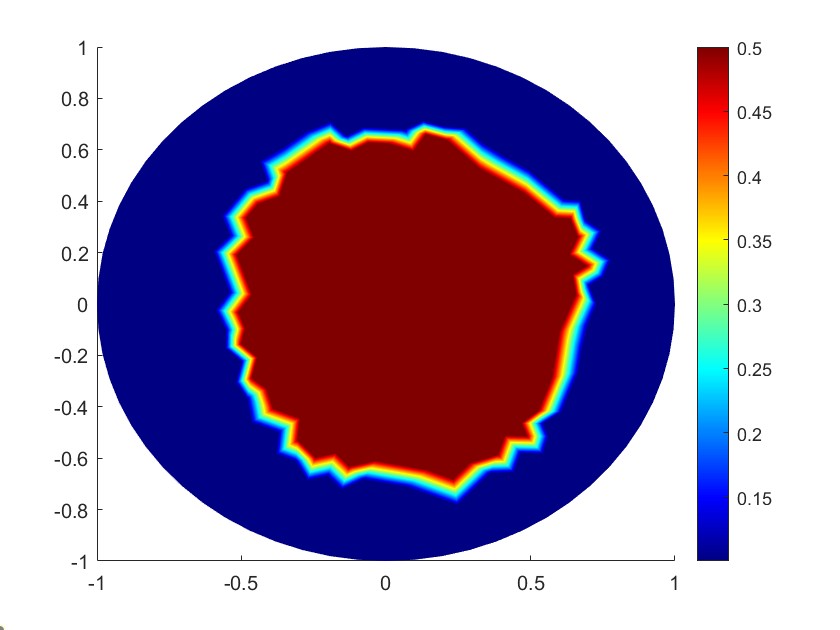}}\\
    {\includegraphics[width=0.22\textwidth]{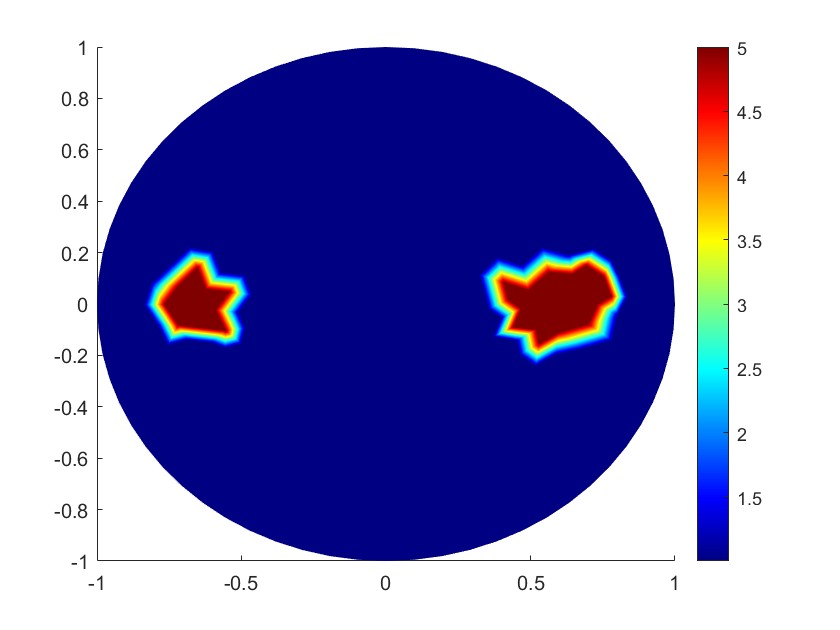}}&
    {\includegraphics[width=0.22\textwidth]{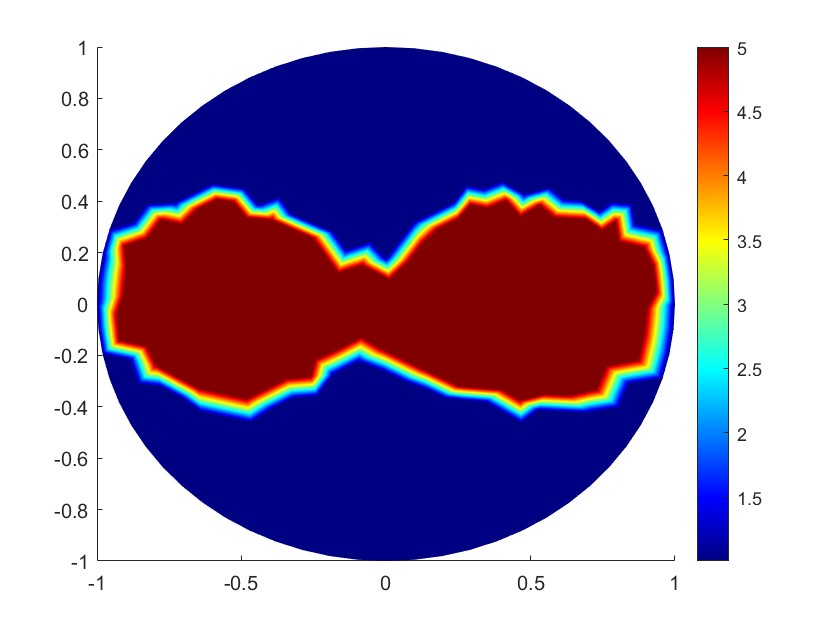}}&
  {\includegraphics[width=0.22\textwidth]{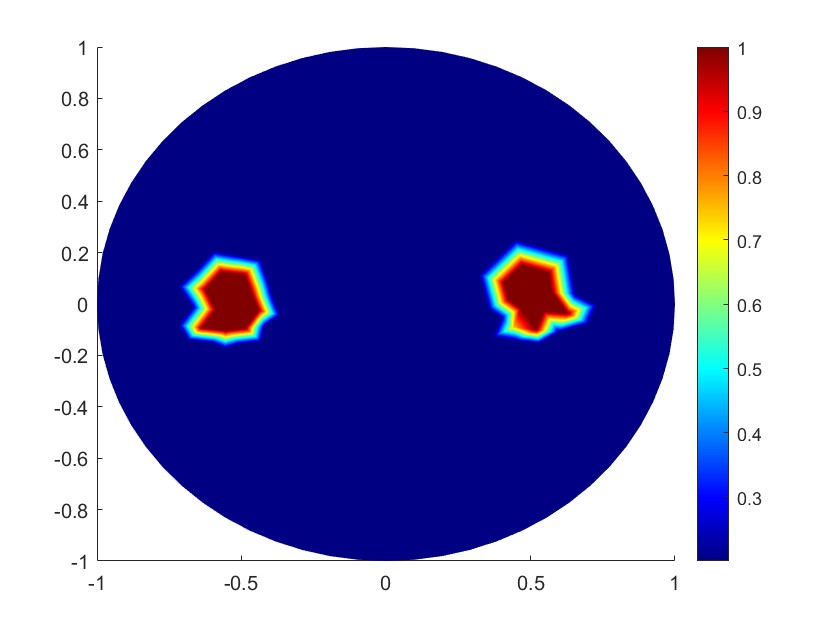}}&
    {\includegraphics[width=0.22\textwidth]{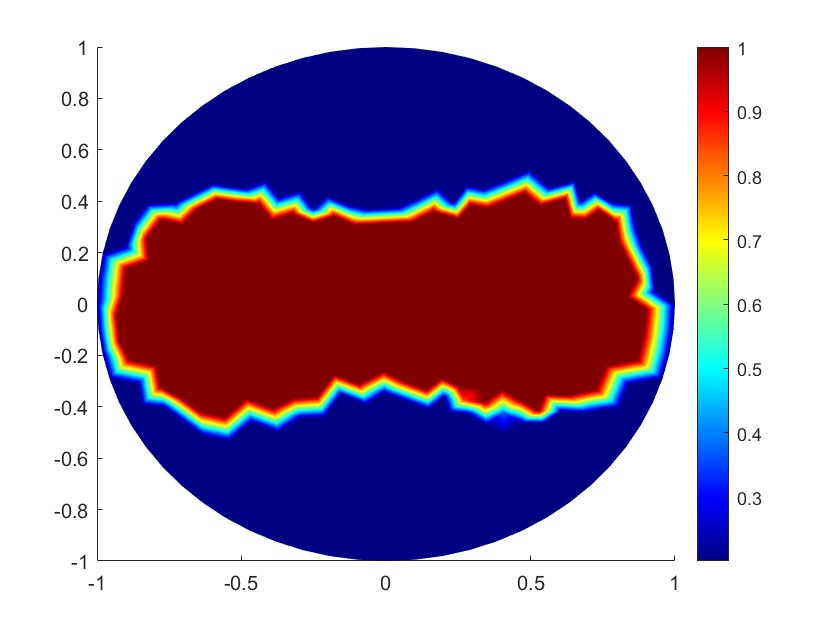}}\\
    {\includegraphics[width=0.22\textwidth]{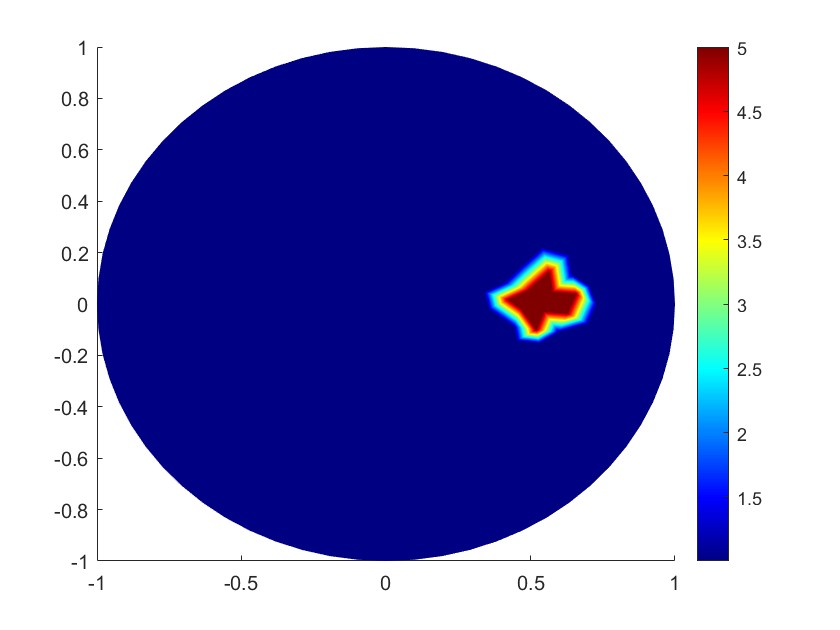}}&
    {\includegraphics[width=0.22\textwidth]{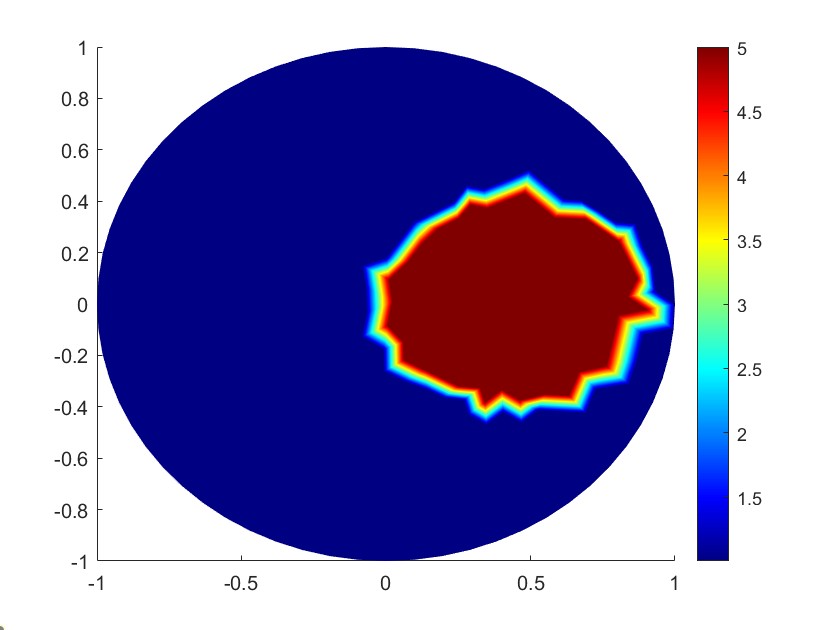}}&
  {\includegraphics[width=0.22\textwidth]{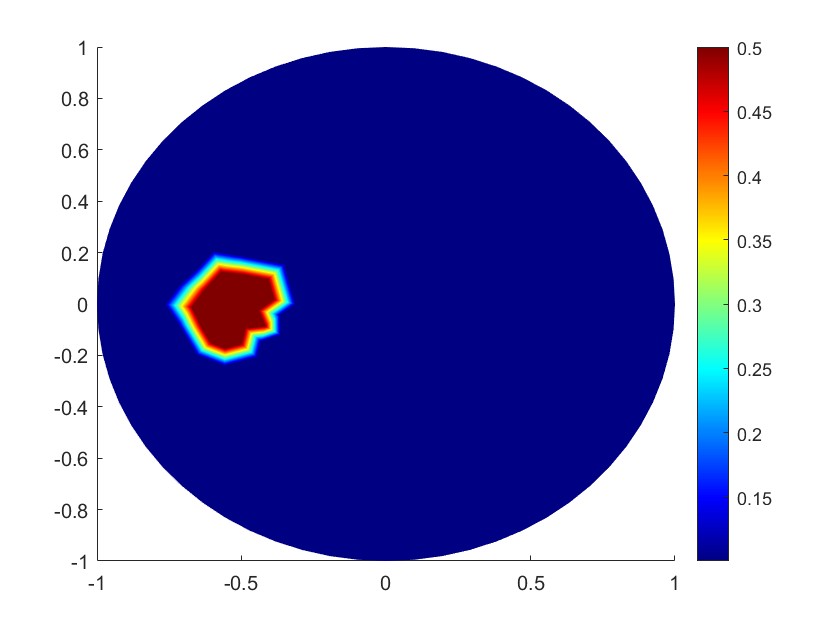}}&
    {\includegraphics[width=0.22\textwidth]{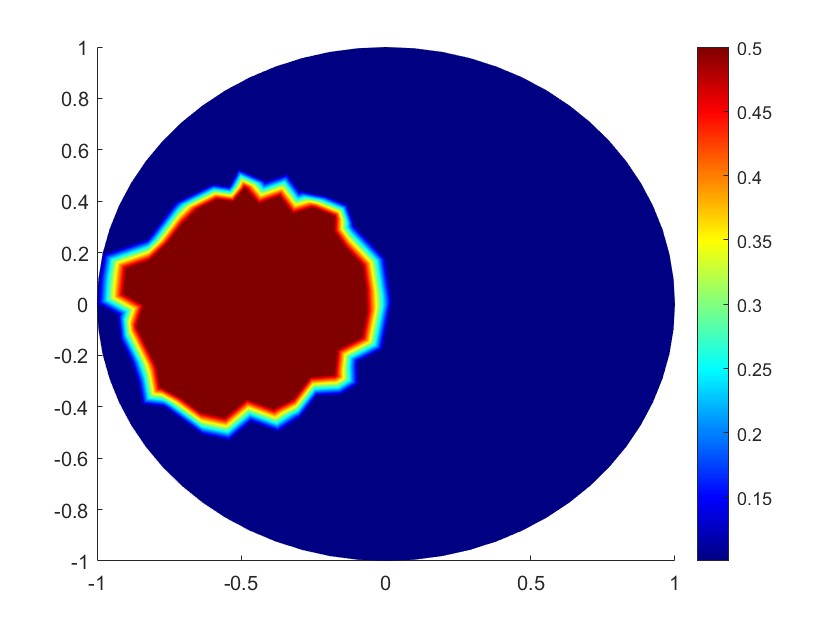}}

      \\ 
      \end{tabular}

      \end{center}
 \caption{Uncertainty quantification in Method 1 (Bi-Level Reconstruction): The rows correspond to the phantom geometries as in Figure \ref{mod1}. The first and third column in each row shows the lower Bayesian credible region ($15\%$) for the diffusive and absorptive regions respectively. The second and the fourth column show the upper credible region ($85\%$) for the diffusive and absorptive regions respectively.}
 \label{UQ}
      \end{figure}

\subsection{Method 2: Continuous Reconstruction}
Reconstruct using $\overline{H_i(u_i)}$ where $\overline{H_i(u_i)}$ refers to the mean of $H_i(u_i)$ over posterior samples of $u_i$. These reconstructions will, by design, not produce bi-level reconstructions. However, they can be used in conjunction with Method 1 to form a better picture of the location of the anomalous regions. For method 2, we also calculate standard error of the samples, {to quantify uncertainty in our estimates}. These have been shown in Figure \ref{mod2}. {We notice that regions of greater uncertainty in the reconstruction lie close to the edge of the phantom reconstructions.}

 \begin{figure}
     \begin{center}
     \begin{tabular}{ c c  c c c c }

       Truth $a$& Recon. $a$ & Error $a$ & Truth $b$&  Recon. $b$  & Error $b$\\ 
{\includegraphics[width=0.14\textwidth]{Images/moon/d_moon_recon_mesh.jpg}}&
{\includegraphics[width=0.14\textwidth]{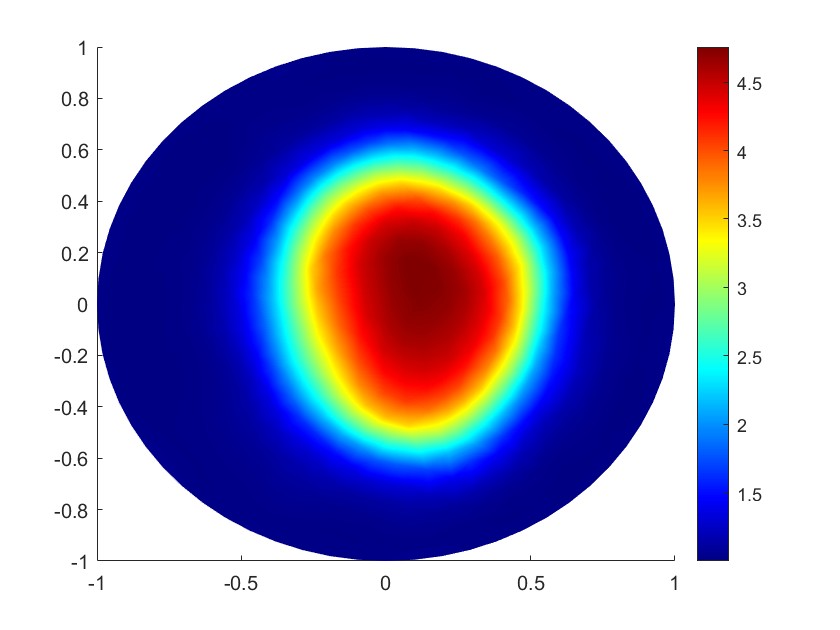}}&
    {\includegraphics[width=0.14\textwidth]{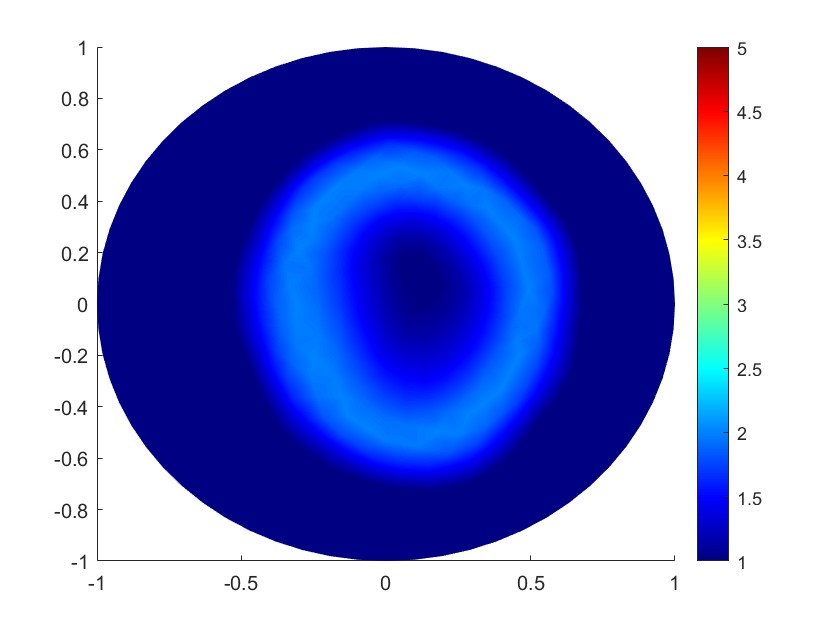}}&
 {\includegraphics[width=0.14\textwidth]{Images/moon/mu_moon_recon_mesh.jpg}}&{\includegraphics[width=0.14\textwidth]{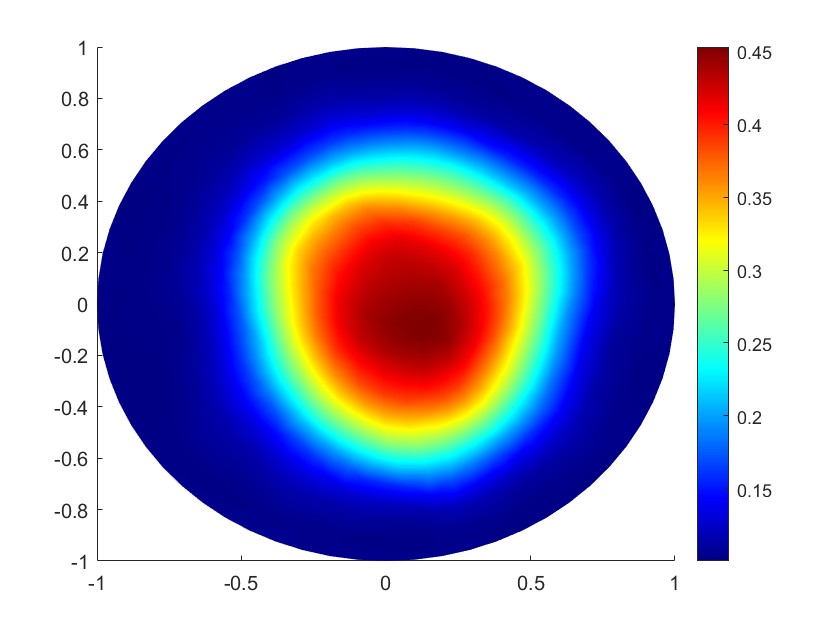}}&
    {\includegraphics[width=0.14\textwidth]{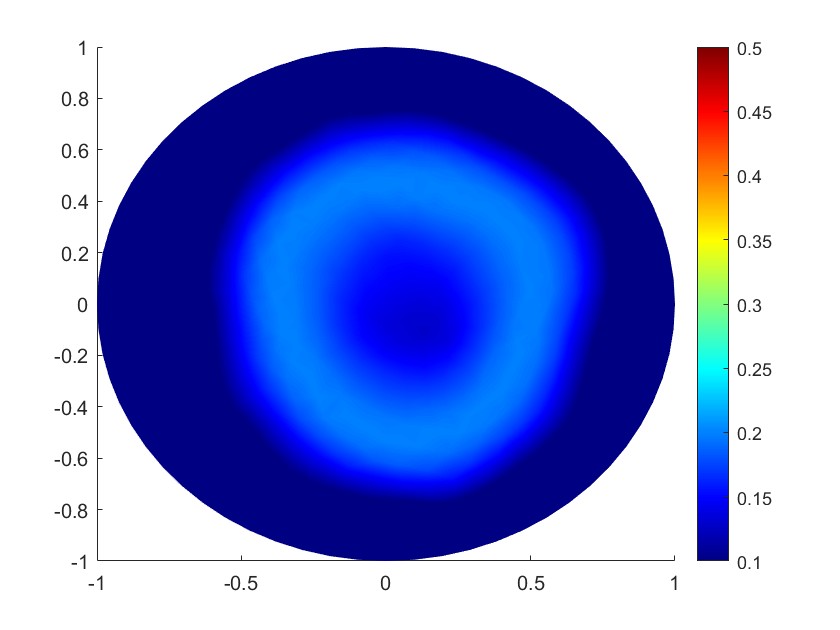}}\\
  {\includegraphics[width=0.14\textwidth]{Images/two/D_2_circ_mesh_t.jpg}} & {\includegraphics[width=0.14\textwidth]{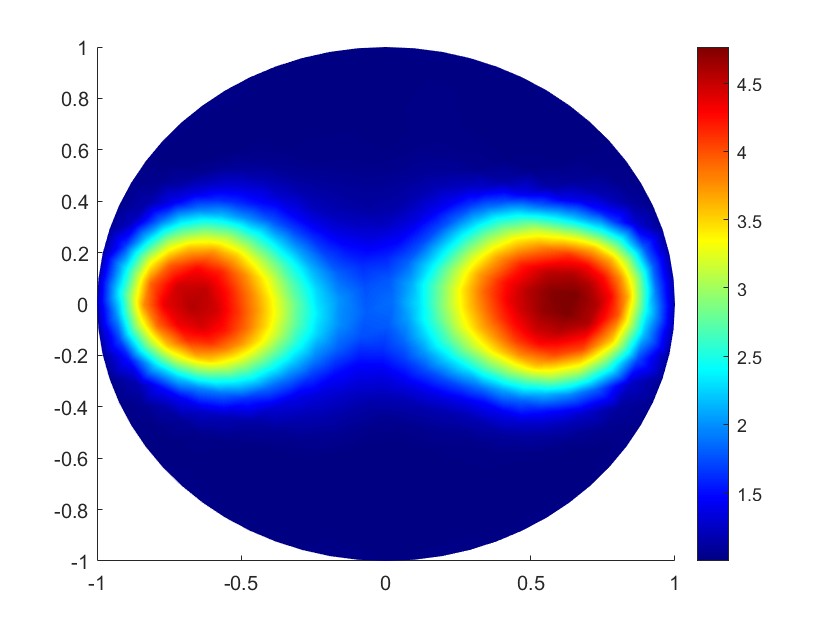}}&
    {\includegraphics[width=0.14\textwidth]{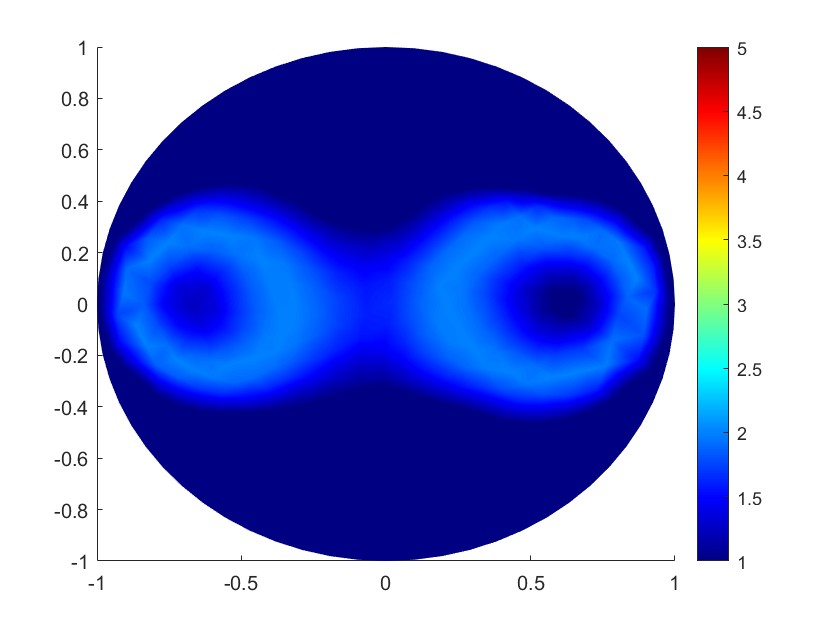}}&
{\includegraphics[width=0.14\textwidth]{Images/two/mu_mesh_two_circ.jpg}} &  {\includegraphics[width=0.14\textwidth]{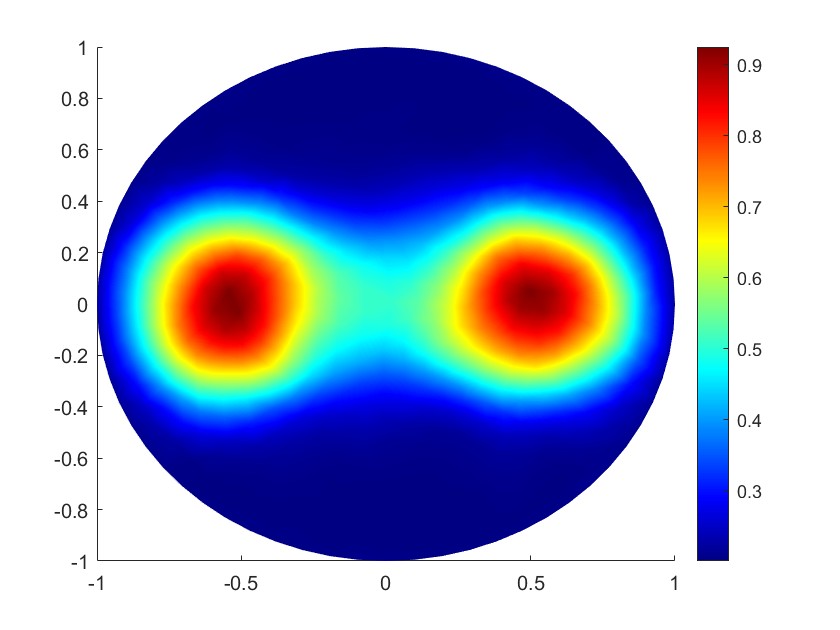}}&
    {\includegraphics[width=0.14\textwidth]{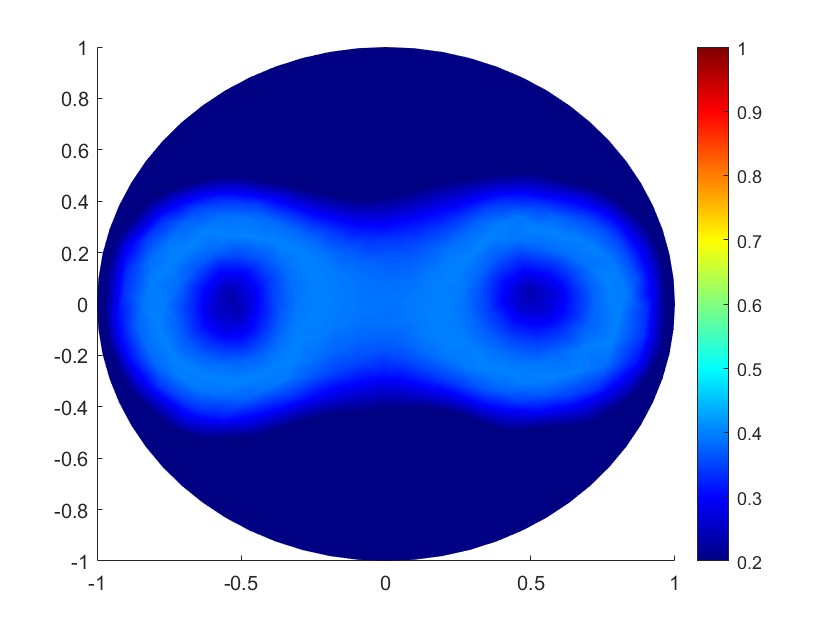}}\\
 {\includegraphics[width=0.14\textwidth]{Images/separate/separate_d_true_mesh.jpg}} &  
{\includegraphics[width=0.14\textwidth]{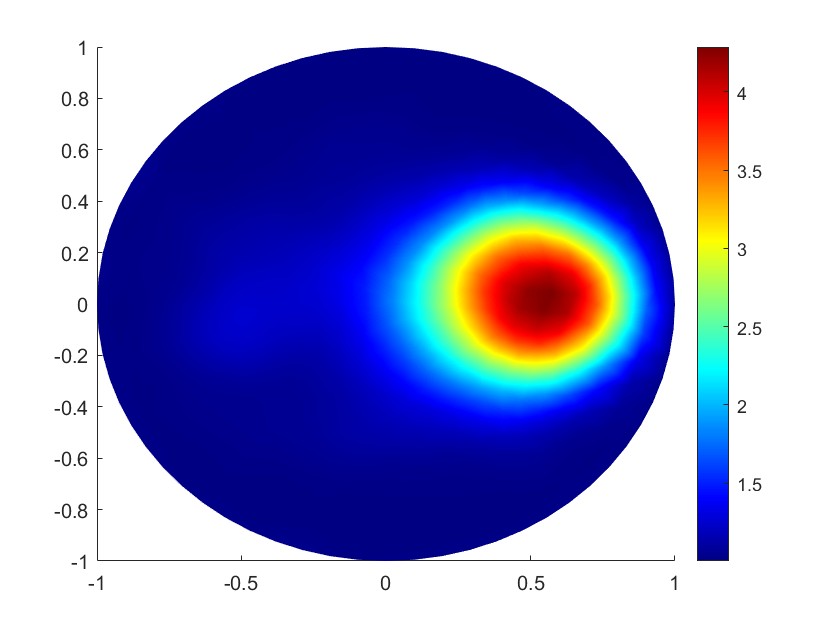}}&
    {\includegraphics[width=0.14\textwidth]{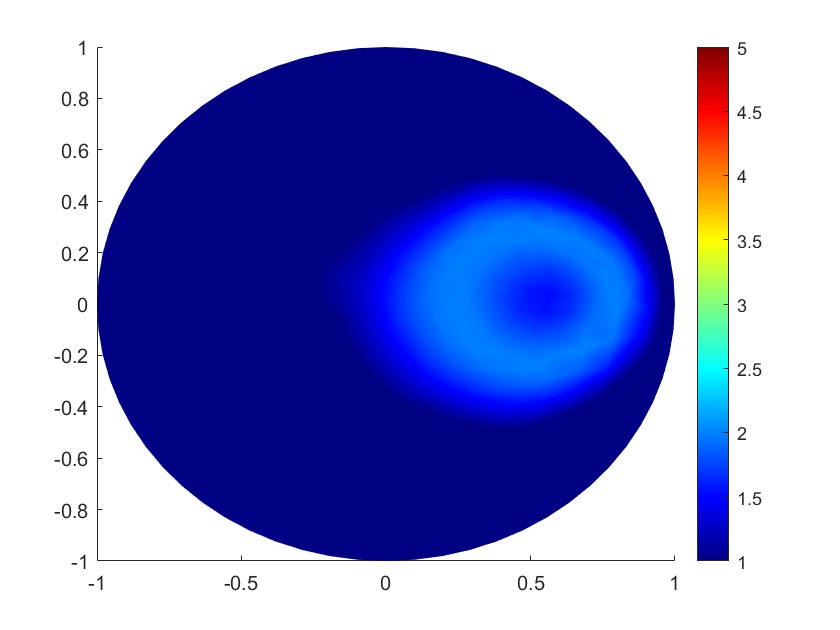}}&
  {\includegraphics[width=0.14\textwidth]{Images/separate/separate_mu_mesh.jpg}}&{\includegraphics[width=0.14\textwidth]{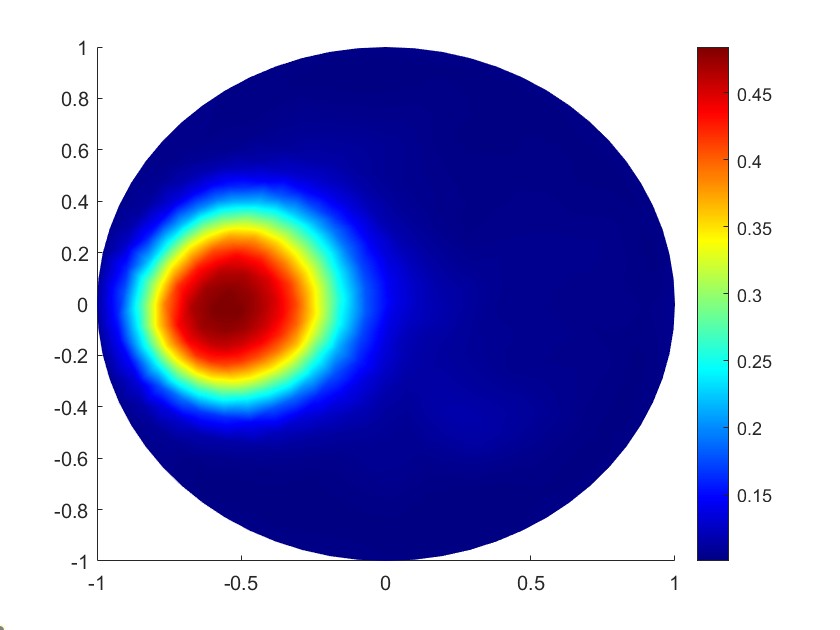}}&
    {\includegraphics[width=0.14\textwidth]{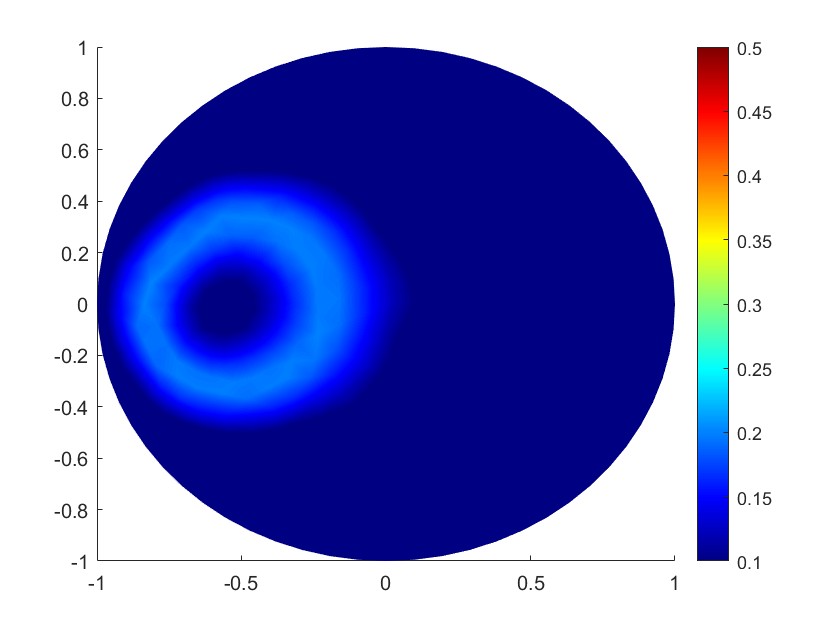}}

      \\ 
      \end{tabular}
\end{center}
      \caption{Reconstructions using Method 2 (Continuous Reconstruction):  The rows correspond to the phantom geometries as shown in Figure \ref{mod1}. In the first and fourth columns, we have the {projections of true diffusive and absorptive phantoms on the FEM basis}. In columns 2 and 5, we have the corresponding reconstructions of the diffusive and absorptive parameters, respectively. In the third and the sixth columns are plots of standard errors. From the standard error plots, it is clear that the greatest variance (and thus, the uncertainty) in posterior samples is around the boundaries of the diffusive/absorptive regions in the phantom. Data was corrupted with $2\%$ relative noise.}
 \label{mod2}
      \end{figure}

 \subsection{Noise-Study}
Now we present results from our noise-study. For method 1, a good measure of accuracy of the reconstruction is to tabulate the symmetric difference (in area) between the reconstructed diffusive (or, absorptive) regions with the true diffusive (respectively, absorptive) regions. In our set-up, we implement a related idea.
Instead of evaluating the symmetric difference, we see the problem as a classification task by computing the accuracy ratio for each finite element for diffusive and absorptive separately. This is done by computing the ratio of finite elements of the reconstruction that are identical to the ground truth. Because the ground truth and the reconstruction are in different mesh sizes, we mapped the ground truth into the reconstruction mesh. 
 
A high value of this ratio (close to $1$) will then indicate a high level of accuracy in the reconstruction. For every chosen noise level, this ratio is evaluated for ten independent experiments (i.e. ten different runs of the MCMC algorithm). We note that due to the significant amount of time that it will take to run an MCMC chain with $300,000$ iterations $10$ times for each noise level, we chose to do the parameter reconstruction by running the MCMC chain for just $100,000$ iterations while throwing away $50,000$ iterations as burn-in. We evaluate the mean accuracy ratio of the ten experiments and the standard deviation (s.d.) for each of the two parameters. This is done for noise levels $1\%-4\%$, and the results for phantoms corresponding to row 2 in Figure \ref{mod1} are tabulated below, see Table \ref{ErrorT}. We do a similar evaluation for method 2 at noise-levels $1\%-4\%$ and use the $L^\infty $-errors incurred in the estimation for such an evaluation. We tabulate the mean $L^\infty$-error for each noise-level, where the mean error for each parameter has been taken over ten independent MCMC experiments; as was done for method 1. In a similar fashion, we tabulate the standard deviation of the errors for each of the two parameters, see Table \ref{meth2T}. The trend reflects the robustness of the reconstruction to small changes in data.
  \begin{table}[!ht]
\centering
 \caption{Accuracy ratio for reconstruction using method 1}
\scriptsize
\begin{tabular}[t]{lcccc}
\hline
Noise level&  mean accuracy for $a(x)$& s.d. for $a(x)$ &  mean accuracy for $b(x)$& s.d. for $b(x)$\\
\hline
$1\%$  &$0.9267$&$0.0084$&$0.8710$&$0.0125$\\
$2\%$  &$0.9240$&$0.0089$&$0.8645$&$0.0312$\\
$3\%$  &$0.8976$&$0.0136$&$0.8484$&$0.0341$\\
$4\%$  &$0.9056$&$0.0139$&$0.8268$&$0.0324$\\

\hline
\label{ErrorT}
\end{tabular}
\end{table}

 \begin{table}[!ht]
\centering
 \caption{$L^\infty$-errors in reconstruction using method 2}
\scriptsize
\begin{tabular}[t]{lcccc}
\hline
Noise level&  mean $L^\infty$-error for $a(x)$& s.d. for $a(x)$ &  mean $L^\infty$-errors for $b(x)$& s.d. for $b(x)$\\
\hline
$1\%$  &$ 3.1513$&$0.1315$&$0.3124$&$0.0163$\\
$2\%$  &$3.2885$&$0.1428$&$0.3112$&$0.0225$\\
$3\%$  &$3.2295 $&$0.1410$&$0.3136$&$0.0273$\\
$4\%$  &$3.2591$&$0.1398$&$0.3220$&$0.0271$\\

\hline
\label{meth2T}
\end{tabular}
\end{table}

\section{Conclusion} In this article, we extend our previous work on reconstruction of absorption in one-parameter DOT (see \cite{Abhi_22b}) to simultaneous reconstruction of both absorption and diffusion parameters from noisy boundary data obtained in the form of Neumann-to-Robin map. The reconstruction is based on a Bayesian formulation of the problem and uses geometric level-set priors which are particularly useful for the case when we consider the parameters to have a piecewise constant behaviour. In this work, at first we establish that the Bayesian inverse problem for the considered CIP is well-posed. We follow the approach proposed in 
 \cite{stuart16b, stuart16a} and extend the analogous result for EIT that was obtained in \cite{stuart_16} to the present case. While the inverse problem considered here is closely related to that of EIT; one major difference between the two is that in this article we seek to reconstruct two  parameters at once (namely, diffusion and absorption) while in EIT, the reconstruction of only one parameter, conductivity, is sought. We implement an MCMC algorithm whereby posterior samples are obtained via a pCN sampling strategy. We illustrate good mixing properties of the MCMC chain by giving traceplots. The posterior samples are then used to reconstruct the parameters. We perform numerical studies and provide several examples of reconstruction. We use two different methods of reconstruction as outlined in the simulation section. We then numerically reconstruct for different phantom geometries with two levels for the parameters, a background diffusive (absorptive) region and a foreground diffusive (absorptive) region. Our numerical studies and theoretical results indicate that the Bayesian level-set inversion procedure is a computationally feasible and theoretically sound method for simultaneous reconstruction of optical parameters in Diffuse Optical Tomography.

\section*{Acknowledgement(s)}
TK would like to acknowledge partial funding from the Center for Trustworthy Artificial Intelligence through Model Risk Management (TAIMing AI) at UNC Charlotte which is funded through the Division of Research, the School of Data Science, and the Klein College of Science.












\bibliographystyle{plain}
\bibliography{references}

\end{document}